\documentclass[12pt, draftclsnofoot, onecolumn]{IEEEtran}

\usepackage{cite}
\usepackage{graphicx,multirow}
\usepackage{amssymb,hhline,enumerate,dsfont}
\usepackage{amsmath}
\usepackage{array}
\usepackage{amsmath,url}
\usepackage{psfrag}
\usepackage[usenames,dvipsnames]{pstricks}
\usepackage{epsfig}
\usepackage{pst-grad} 
\usepackage{pst-plot} 

\newtheorem{thm}{Theorem}
\newtheorem{lem}{Lemma}

\newtheorem{cor}{Corollary}
\newtheorem{proof}{Proof}

\IEEEoverridecommandlockouts

\begin{document}
%
\title{Linear-Complexity Overhead-Optimized Random Linear Network Codes}
\author{Kaveh Mahdaviani,~\IEEEmembership{Student Member, IEEE}, Raman Yazdani,~\IEEEmembership{Member, IEEE}, Masoud Ardakani,~\IEEEmembership{Senior Member, IEEE}
\thanks{This paper was presented in part at
the 2012 International Symposium on Network Coding (NetCod'12), Cambridge, MA, USA \cite{Mahdaviani12}, and 2013 International Symposium on Network Coding (NetCod'13), Calgary, Canada \cite{Mahdaviani13}.}
}

%

\maketitle

\begin{abstract}
Sparse random linear network coding (SRLNC) is an attractive technique proposed in the literature to reduce the decoding complexity of random linear network coding. Recognizing the fact that the existing SRLNC schemes are not efficient in terms of the required reception overhead, we consider the problem of designing overhead-optimized SRLNC schemes. To this end, we introduce a new design of SRLNC scheme that enjoys very small reception overhead while maintaining the main benefit of SRLNC, i.e., its linear encoding/decoding complexity. We also provide a mathematical framework for the asymptotic analysis and design of this class of codes based on density evolution (DE) equations. To the best of our knowledge, this work introduces the first DE analysis in the context of network coding. Our analysis method then enables us to design network codes with reception overheads in the order of a few percent. We also investigate the finite-length performance of the proposed codes and through numerical examples we show that our proposed codes have significantly lower reception overheads compared to all existing linear-complexity random linear network coding schemes.
\end{abstract}


%
\IEEEpeerreviewmaketitle

\section{Introduction} \label{sec:introduction}

Soon after the introduction of its basic concept in \cite{NetCoding}, network coding was accepted as a promising technique for multicast and attracted a lot of attention in the research community. As opposed to conventional packet networks where intermediate nodes can only store and forward the incoming packets, in network coding the intermediate nodes can also combine the incoming packets to form (encode) an outgoing packet. Later, the idea of linearly combining the incoming packets was introduced in \cite{LinNetCoding} and extended in \cite{AlgNetCoding} by using an algebraic approach. Also by applying random coefficients, random linear network coding (RLNC) \cite{Benefits, RandNetCoding} was shown to be sufficient for achieving zero reception overhead with failure probability arbitrarily close to zero. As a result, network coding became an attractive technique for multicast over networks with random topology.

In RLNC, the source node and all the intermediate nodes of the network encode the data packets by forming random linear combinations of them. The receivers then wait to receive enough encoded packets, in other words enough linear combinations of the information packets, such that they can form full rank systems of linear equations. Each receiver can now decode the information packets by solving the system of linear equations corresponding to the set of received packets. It is shown in \cite{Benefits, RandNetCoding} that by using RLNC with a sufficiently large code alphabet $q$, it is possible to achieve zero reception overhead\footnote{In this paper, the reception overhead is defined as the difference between the number of received packets required for successful decoding and the number of information packets divided by the number of information packets.} with failure probability arbitrary close to zero. The encoding complexity of RLNC for a block of $K$ information packets each with $s$ symbols is $O(Ks)$ operations per coded packet where the operations are done in $\mathrm{GF}(q)$. The complexity of decoding then scales as $O(K^2+Ks)$ per information packet which becomes impractical when the block size $K$ is moderate to large.


To reduce the decoding complexity of network coding, the idea of fragmenting the information packet blocks into distinct \emph{generations} is proposed in \cite{Practical}. This way, random linear combinations are formed only within each generation. This makes the final linear equation system solvable locally within each generation and thus sparse. This technique, however, requires a large number of control messages to be exchanged between the nodes to combat the problem of \emph{rare blocks} and \emph{block reconciliation} \cite{bharambe06}. To avoid this, a method called \emph{sparse} RLNC (SRLNC) is proposed in \cite{Efficient_Methods} which uses a simple random schedule for selecting which generation to transmit at any time. This method reduces the encoding complexity to $O(gs)$ per coded packet and the decoding complexity to $O(g^2+gs)$ per information packet, where $g$ denotes the number of information packets in each generation. It has been shown that to keep the SRLNC tractable in most of the practical settings the generation size should be upper bounded by a small constant (e.g. 512 for typical notebook computers \cite{Uusee}). Throughout this paper we limit our discussion to the case where all generations are of a small constant size $g$, which does not scale with the information block length $K$. This complexity is practically feasible if $g$ is not very large, making SRLNC an attractive solution for multicast. Unfortunately, the reception overhead under this scheme is affected by the \emph{curse of coupon collector} phenomenon \cite{Newman_Shepp , Erdos_Renyi_Coupon}, and thus even for very large alphabet size or number of information packets, the reception overhead does not vanish. In fact, the reception overhead grows with $K$ as $O(\log K)$ \cite{RandomAnnex}. Consequently, a trade-off is raised between reception overhead and complexity in SRLNC.

In general, the large reception overhead in SRLNC comes from two sources. The first and major source is random scheduling. More specifically, for all generations to become full rank, due to random scheduling, some generations will receive significantly more than $g$ packets resulting in a large reception overhead. The second source of reception overhead is the possibility of receiving linearly dependent combinations of the packets. The probability of receiving linearly dependant equations can be arbitrarily reduced by increasing the field size $q$ of the code alphabet \cite{RandNetCoding,Silva_Overlapping}. Another solution recently proposed in \cite{silva12} is to do pre-coding using maximum rank distance codes which is quite effective even for very small field sizes.


Knowing that SRLNC is complexity-efficient, there have been several attempts to decrease its reception overhead \cite{Efficient_Methods,Silva_Overlapping,Heidarzadeh,silva12,RandomAnnex,tang12,tang13}. For this purpose, the idea of using an outer code is introduced in \cite{Efficient_Methods}. In this method, an outer code which is considered as a separate block is applied to SRLNC. At the receiver, the outer decoder waits for the recovery of $1-\delta$ fraction of the generations for some small predefined $\delta$ and then participates in the decoding to recover the remaining $\delta$ fraction of the generations. This method is capable of reducing the reception overhead to a constant, independent of $K$. However, this scheme is still wasteful in terms of the reception overhead since it ignores the received packets pertaining to the $\delta$ fraction of the generations. Furthermore, waiting to receive enough packets to recover $1-\delta$ fraction of the generations when $\delta$ is small leads to a high probability of receiving more than $g$ packets in many generations. As a result, the reception overhead is considerably large even for infinite block lengths.

In \cite{Silva_Overlapping, Heidarzadeh} the idea of overlapping generations, where some packets are shared among different generations, is proposed. This overlap reduces the reception overhead of SRLNC since generations can help each other in the decoding process. Another \emph{overlapped} SRLNC scheme called \emph{Random Annex} codes \cite{RandomAnnex} proposes \emph{random} sharing of the packets between any two generations. Furthermore, combination of overlapping and outer coding (called expander chunked (EC) codes) is proposed in \cite{tang12,tang13}. In this scheme, expander graphs are used to form overlapped generations. By establishing an upper bound on the reception overhead and careful choice of parameters, it has been shown that the proposed scheme of \cite{tang12,tang13} outperforms all other previously existing overlapping schemes. However, as it will be shown later, the SRLNC designed based on EC codes analysis is equivalent to a special (but not the optimal) case of our proposed Gamma codes.

Another recently proposed low-complexity RLNC technique is called batched sparse (BATS) coding \cite{yang11,yang12}. BATS codes generalize the idea of fountain codes \cite{LT,Raptor} to a network with packet loss. Using BATS codes, the buffer requirement at intermediate nodes become independent of the information packets block size in tree networks \cite{yang12}. Furthermore, it is shown that these codes perform near the capacity of the underlying network in certain cases \cite{yang12}. Nevertheless, design of BATS codes requires the knowledge of the end-to-end transformation of the packets in the network which is not always available.

One important difference between BATS codes and SRLNC is that BATS codes perform RLNC inside variable-size subsets of information packets called \emph{batches} whereas SRLNC schemes perform RLNC inside fixed-size subsets of information packets, i.e., generations. For each batch in BATS codes, first a degree $d_i$ is determined by sampling from a degree distribution. Then, $d_i$ information packet is chosen uniformly at random from all information packets. Next, a fixed number $M$ of output packets (forming a batch) are encoded by performing RLNC inside the chosen subset. Then, the same process is repeated to form the next batch. The receiver needs to receive enough number of batches to decode all information packets. In BATS codes, each information packet can participate in multiple batches which can be seen as overlaps between batches. Thus, BATS codes and overlapping SRLNC are similar in imposing dependence between batches/generations.

A key observation which will lead to our proposed network coding scheme is that the overlap between different generations in overlapped SRLNC can be seen as having a repetition outer code acting on the common packets from overlapping generations. Thus, overlapped SRLNC can be seen as a special case of SRLNC with outer code. In overlapped SRLNC, on the contrary to the separate outer coding of \cite{Efficient_Methods}, there is no need to wait for the recovery of a large fraction of the generations before the repetition outer code can participate in the decoding. This can potentially reduce the reception overhead compared to the scheme of \cite{Efficient_Methods}. This point of view then leads to the idea of allowing the outer code to participate in the decoding, but not limiting the outer code to a repetition code. This in turn generates a host of new questions, some of which are answered in this work. For example, a major question is how one can design an outer code which provides minimum reception overhead. To the best of our knowledge, no general analysis and design technique for SRLNC with an outer code exists in the literature. The analysis methods presented in \cite{Heidarzadeh_Analysis,RandomAnnex,tang13} either assume specific network structures or specific coding schemes such as overlapping schemes and thus cannot be used to design outer coded SRLNC in a general way\footnote{In this work, the only constraint on the outer code is that we consider the class of linear outer codes that choose their variable nodes uniformly at random, which we refer to them as \emph{random linear outer codes}. This constraint simplifies the analysis and design of optimal codes. As will be revealed in the results section, despite the mentioned constraint, the optimal design achieves asymptotic overheads as small as $2\%$.}.

\subsection{Main idea and summary of contributions}

In this work, we propose a solution to the problem of designing low-overhead linear-complexity SRLNC with a random linear outer code. For this purpose, we introduce a new family of low-overhead linear-complexity network codes, called \emph{Gamma network codes}. In Gamma network codes, SRLNC with outer code is considered in a more general way, i.e., the outer code is not limited to a simple repetition outer code. Also, Gamma network codes do not rely on a large portion of generations being recovered before getting the outer code involved in the decoding. We then develop an analytical framework based on density evolution equations \cite{DensityEvolution} to investigate the impact of the outer code parameters on the average reception overhead. The importance of such framework is that it can be used both for (i) finding the limits on the performance measures of SRLNC with random linear outer code such as the minimum achievable reception overhead, (ii) track the decoding process, and (iii) to analytically design optimal codes.

While similar to \cite{Efficient_Methods} an outer code is suggested here, the design of Gamma network codes has major differences with that of \cite{Efficient_Methods}: (i) Unlike \cite{Efficient_Methods}, where the outer code has to wait for a large fraction of the generations to be recovered, here the outer code can participate in the decoding as soon as a single generation is recovered. In other words, outer decoding is done jointly with solving the linear equation systems instead of separate decoding used in \cite{Efficient_Methods}. (ii) In contrast to \cite{Efficient_Methods}, the received packets belonging to non-full rank generations are not ignored. (iii) Our outer codes are designed to have the ability of actively participating in the decoding when the fraction of known packets is much smaller than the code rate. As we will show later, the reception overhead of Gamma network codes is significantly smaller than that of \cite{Efficient_Methods}.

Gamma network codes are built based on the following facts/results: (1) Every received packet whose corresponding linear combination is linearly independent with those of all other received packets is innovative and must be used in the decoding process. (2) Assuming the field size of the code alphabet is large enough, before receiving enough packets to form a small number of full-rank generations, all received packets are linearly independent with high probability. (3) It is possible to design an outer code capable of successful decoding, based on receiving enough packets to have only a small fraction of full rank generations. Details of this code design is provided in Section \ref{sec:optimization}.

In summary, our solution works in the following way. Accepting an optimally small reception overhead, we continue receiving packets until a small fraction of the generations is full rank. Next, the carefully designed outer code comes to help to decode all other generations through providing enough information about the packets in the remaining generations to remove the rank deficiency in their corresponding linear equation systems. This will be done in an iterative decoding process alternating between the outer code and the SRLNC. Since nearly all received packets are used in the decoding process, the outer code does not introduce an excess overhead. The key to our finding is an intermediate performance analysis (i.e., density evolution equations) of SRLNC with outer code.

Our contributions are summarized as follows: (i) We introduce a new class of linear-complexity random linear network codes called Gamma network codes. This design is based on integrating a carefully designed outer code into SRLNC. Our design enables joint decoding of the outer and the SRLNC at the receivers and is shown to outperform all other existing linear-complexity random linear network codes. (ii) We derive density evolution equations for the asymptotic performance analysis of Gamma network codes. (iii) Using the asymptotic analysis, we propose an optimization technique to design optimized Gamma network codes with very small reception overheads. (iv) Finite-length performance of these codes are also evaluated and some methods to improve their performance are presented. We also compare our results with those of overlapping SRLNC schemes \cite{Silva_Overlapping, RandomAnnex,Efficient_Methods,tang13}. We will show that Gamma network codes are capable of reducing the reception overhead compared to all the existing linear-complexity random linear network coding schemes.

In our analysis, we assume that as long as less than g packets are received in a generation, these packets are linearly independent. This can be due to using a sufficiently large q or using methods of \cite{silva12}. The assumption is primarily made to prevent unnecessary complications and to be consistent with the convention in the literature \cite{RandomAnnex,Silva_Overlapping}. We study the effect of $q$ on the performance of Gamma network codes numerically in Section \ref{subsec:results_B}.

The rest of this paper is organized as follows. In the next section, we describe the encoding and decoding structure of the proposed Gamma network codes. Section~\ref{sec:analysis} Describes the asymptotic analysis of the performance of the proposed codes through introducing the decoding evolution chart. Next, we will propose an optimization technique to design Gamma network codes with minimum reception overhead in Section~\ref{sec:optimization}. Section~\ref{sec:results} contains numerical results and discussions. Moreover as an example of the applications of decoding evolution chart we will also propose techniques to design Gamma network codes which are suitable for achieving very small decoding failure probability in finite block lengths. Section~\ref{sec:improved_design} reviews the encoding procedure and suggests some improved designs. Finally, Section~\ref{sec:conclusion} concludes the paper.

\section{The Proposed Coding Scheme} \label{Sec:model}

\subsection{Network model} \label{sec:network_model}
Similar to \cite{RandomAnnex, Silva_Overlapping,tang13}, in this paper we consider the transmission of a file consisting of information packets from a source to a destination over a unicast link. The network structure is assumed to be dynamic with diverse routing, unknown and variable packet loss, and with random processing times at the intermediate nodes. It is further assumed that random linear combining is performed at the intermediate nodes on the available packets within each generation. As a result, the destination receives a random subset of the random linear combinations of the transmitted packets and is supposed to recover the information packets.

\subsection{Encoding}\label{sec:encoding}
The encoding process of Gamma network codes is done in two steps. In the first step, a file consisting of $K$ information packets, each having $d$ symbols in $\mathrm{GF}(q)$ is encoded via a linear outer code\footnote{As we will show in Section~\ref{sec:convergence}, a pre-code can also be helpful. To avoid complications, we do not discuss this here and leave it to Section~\ref{sec:convergence}.} $\mathcal{C}$ of rate $R$ giving rise to a block of $N$ \emph{outer coded} packets where $R=K/N$. These $N$ outer coded packets are partitioned into $n=\lceil\frac{N}{g}\rceil$ distinct generations, where $\lceil x\rceil$ is the smallest integer larger than or equal to $x$. In this work, without loss of generality we assume that $N$ is a multiple of $g$, where $g$ denotes the number of packets in each generation.

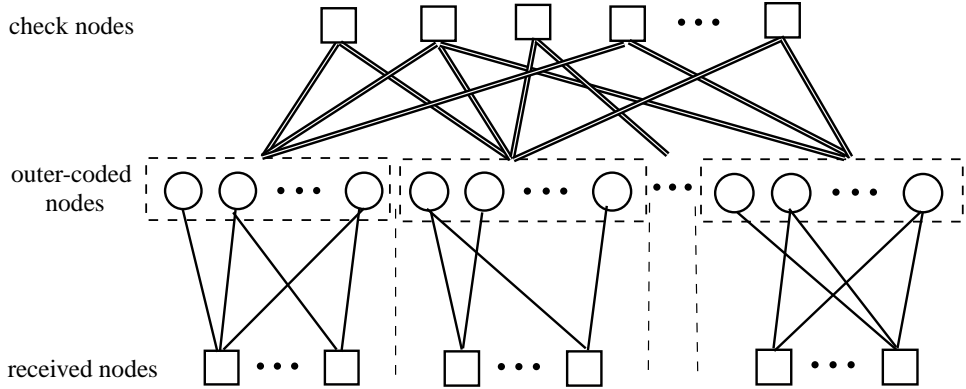
\begin{figure}
\centering
\scalebox{.8} 
{
\begin{pspicture}(0,-3.1998413)(15.697812,3.1998413)
\pscircle[linewidth=0.04,dimen=outer](2.7297537,0.05859398){0.32262176}
\pscircle[linewidth=0.04,dimen=outer](3.6356175,0.05859398){0.32262176}
\pscircle[linewidth=0.04,dimen=outer](5.7345705,0.05859398){0.32262176}
\psdots[dotsize=0.12](4.351425,0.06262501)
\psdots[dotsize=0.12](4.6607447,0.06262501)
\psdots[dotsize=0.12](4.9479694,0.06262501)
\pscircle[linewidth=0.04,dimen=outer](6.816589,0.04676071){0.33524695}
\pscircle[linewidth=0.04,dimen=outer](7.7330594,0.04676071){0.33524695}
\pscircle[linewidth=0.04,dimen=outer](9.856589,0.04676071){0.33524695}
\psdots[dotsize=0.12](8.4484005,0.054780975)
\psdots[dotsize=0.12](8.761342,0.054780975)
\psdots[dotsize=0.12](9.05193,0.054780975)
\pscircle[linewidth=0.04,dimen=outer](11.88691,0.021255901){0.33857176}
\pscircle[linewidth=0.04,dimen=outer](12.832068,0.021255901){0.33857176}
\pscircle[linewidth=0.04,dimen=outer](15.022068,0.021255901){0.33857176}
\psdots[dotsize=0.12](13.57697,0.039386142)
\psdots[dotsize=0.12](13.899707,0.039386142)
\psdots[dotsize=0.12](14.199391,0.039386142)
\psdots[dotsize=0.12](10.597813,0.12015877)
\psdots[dotsize=0.12](10.857813,0.12015877)
\psdots[dotsize=0.12](11.117812,0.12015877)
\psframe[linewidth=0.03,linestyle=dashed,dash=0.16cm 0.16cm,dimen=outer](6.1778126,0.6201588)(2.0978124,-0.44181028)
\psframe[linewidth=0.03,linestyle=dashed,dash=0.16cm 0.16cm,dimen=outer](10.437813,0.60015875)(6.3178124,-0.46332794)
\psframe[linewidth=0.03,linestyle=dashed,dash=0.16cm 0.16cm,dimen=outer](15.697812,0.60015875)(11.317813,-0.46797958)
\psdots[dotsize=0.12](11.009708,2.864938)
\psdots[dotsize=0.12](11.293425,2.864938)
\psdots[dotsize=0.12](11.577141,2.864938)
\psline[linewidth=0.04cm,doubleline=true,doublesep=0.02,doublecolor=white](5.2862716,2.505852)(4.0578127,0.6201588)
\psline[linewidth=0.04cm,doubleline=true,doublesep=0.02,doublecolor=white](5.3378124,2.5201588)(8.157812,0.60015875)
\psline[linewidth=0.04cm,doubleline=true,doublesep=0.02,doublecolor=white](6.9378123,2.5001588)(4.047693,0.63179314)
\psline[linewidth=0.04cm,doubleline=true,doublesep=0.02,doublecolor=white](6.9978123,2.4601588)(8.177813,0.58015877)
\psline[linewidth=0.04cm,doubleline=true,doublesep=0.02,doublecolor=white](6.9978123,2.5001588)(13.717813,0.60015875)
\psline[linewidth=0.04cm,doubleline=true,doublesep=0.02,doublecolor=white](8.577812,2.6001587)(8.197812,0.6201588)
\psline[linewidth=0.04cm,doubleline=true,doublesep=0.02,doublecolor=white](8.597813,2.5601587)(10.777813,0.66015875)
\psline[linewidth=0.04cm,doubleline=true,doublesep=0.02,doublecolor=white](10.057813,2.5001588)(4.0378127,0.6201588)
\psline[linewidth=0.04cm,doubleline=true,doublesep=0.02,doublecolor=white](10.137813,2.5401587)(13.817813,0.6401588)
\psline[linewidth=0.04cm,doubleline=true,doublesep=0.02,doublecolor=white](12.657812,2.6201587)(13.817813,0.6201588)
\psline[linewidth=0.04cm,doubleline=true,doublesep=0.02,doublecolor=white](12.657812,2.5801587)(8.217813,0.58015877)
\psframe[linewidth=0.04,dimen=outer](5.617114,3.1161153)(5.0080285,2.50703)
\psframe[linewidth=0.04,dimen=outer](7.27576,3.1440241)(6.666675,2.5349386)
\psframe[linewidth=0.04,dimen=outer](8.847109,3.1719327)(8.238024,2.5628471)
\psframe[linewidth=0.04,dimen=outer](10.418458,3.1440241)(9.809373,2.5349386)
\psframe[linewidth=0.04,dimen=outer](12.993725,3.1998413)(12.384639,2.590756)
\psframe[linewidth=0.04,dimen=outer](3.6660845,-2.5806284)(3.0668716,-3.1798413)
\psframe[linewidth=0.04,dimen=outer](5.661874,-2.5806284)(5.0626607,-3.1798413)
\psdots[dotsize=0.12](4.009937,-2.8538456)
\psdots[dotsize=0.12](4.2950497,-2.8538456)
\psdots[dotsize=0.12](4.5801625,-2.8538456)
\psframe[linewidth=0.04,dimen=outer](7.663823,-2.5903013)(7.054283,-3.1998413)
\psframe[linewidth=0.04,dimen=outer](9.697941,-2.5903013)(9.088401,-3.1998413)
\psdots[dotsize=0.12](8.01546,-2.8629901)
\psdots[dotsize=0.12](8.306047,-2.8629901)
\psdots[dotsize=0.12](8.596636,-2.8629901)
\psframe[linewidth=0.04,dimen=outer](12.85245,-2.5642562)(12.236865,-3.1798413)
\psframe[linewidth=0.04,dimen=outer](14.950239,-2.5642562)(14.334655,-3.1798413)
\psdots[dotsize=0.12](13.228128,-2.8312907)
\psdots[dotsize=0.12](13.527813,-2.8312907)
\psdots[dotsize=0.12](13.827497,-2.8312907)
\psline[linewidth=0.02cm,linestyle=dashed,dash=0.16cm 0.16cm](6.2578125,-0.23984122)(6.2578125,-2.9998412)
\psline[linewidth=0.02cm,linestyle=dashed,dash=0.16cm 0.16cm](11.237812,-0.19984123)(11.277813,-2.9398413)
\psline[linewidth=0.02cm,linestyle=dashed,dash=0.16cm 0.16cm](10.477813,-0.13984123)(10.457812,-2.9198413)
\psline[linewidth=0.04cm](2.7178125,-0.23984122)(3.2978125,-2.6198413)
\psline[linewidth=0.04cm](3.5978124,-0.25984123)(3.3378124,-2.6198413)
\psline[linewidth=0.04cm](5.6578126,-0.21984123)(5.3578124,-2.6198413)
\psline[linewidth=0.04cm](3.5578125,-0.29984123)(5.2778125,-2.5798411)
\psline[linewidth=0.04cm](5.7328863,-0.23096144)(3.3978126,-2.5798411)
\psline[linewidth=0.04cm](9.767224,-0.19063903)(9.454283,-2.5903013)
\psline[linewidth=0.04cm](6.7778125,-0.23984122)(9.417812,-2.5598412)
\psline[linewidth=0.04cm](7.6978126,-0.29984123)(7.3778124,-2.5598412)
\psline[linewidth=0.04cm](6.8378124,-0.25984123)(7.3578124,-2.5798411)
\psline[linewidth=0.04cm](11.937813,-0.29984123)(14.577812,-2.5798411)
\psline[linewidth=0.04cm](15.017813,-0.29984123)(14.597813,-2.5598412)
\psline[linewidth=0.04cm](14.994023,-0.28105536)(12.550444,-2.6042562)
\psline[linewidth=0.04cm](12.797812,-0.23984122)(14.577812,-2.5398412)
\psline[linewidth=0.04cm](12.817813,-0.25984123)(12.537812,-2.5598412)
\usefont{T1}{ptm}{m}{n}
\rput(0.87546873,2.8101587){check nodes}
\usefont{T1}{ptm}{m}{n}
\rput(0.8721875,0.35015878){outer-coded}
\usefont{T1}{ptm}{m}{n}
\rput(1.0604688,-2.8898413){received nodes}
\usefont{T1}{ptm}{m}{n}
\rput(0.93046874,-0.12984122){nodes}
\end{pspicture}
}
\caption{The graphical representation for a Gamma network code with check nodes, outer coded nodes, and received nodes corresponding to outer code's check equations, outer coded packets, and received packets, respectively. Each group of outer coded nodes constituting a generations is separated by a dashed box. The edges of the outer code's check nodes are in fact hyper edges connecting dense linear combinations of the outer coded packets in the corresponding generation to the check node. The degree of outer code's check nodes is defined as the number of generations connected to it. For example, the degree of the leftmost check node is $2$.}\label{fig:decoding_graph}
\end{figure}


The structure of the linear outer code $\mathcal{C}$ requires some explanation. Fig.~\ref{fig:decoding_graph} shows the graphical representation of a Gamma network code. As the figure shows, in contrast to the check nodes of a conventional linear code which represent parity-check equations imposed on the connected encoded packets, check nodes in $\mathcal{C}$ represent parity-check equations imposed on dense random linear combinations\footnote{A linear combination is called dense when the coefficients are non-zero with high probability. When the coefficients are drawn uniformly at random from $\mathrm{GF}(q)$ the linear combination will be dense.} of the encoded packets of the connected generations. For example, the parity-check equation of the check node $\mathsf{c}$ is given by
\begin{equation}\label{eq:generation_based_checks}
\sum_{i\in\mathcal{N}(\mathrm{c})}\sum_{j=1}^g\alpha_j^{(i)}u_j^{(i)}=0,
\end{equation}
where $\mathcal{N}(\mathrm{c})$ denotes the set of generations connected to $\mathrm{c}$, $\alpha_j^{(i)}$ is the random coefficient of the $j$th outer coded packet from the $i$th generation chosen uniformly at random from $\mathrm{GF}(q)$, and $u_j^{(i)}$ denotes the $j$th outer coded packets from the $i$th generation. For reasons that will be revealed later in Section~\ref{sec:density_evolution}, we characterize the outer code $\mathcal{C}$ by a generating polynomial $P(x)=\sum_{i=2}^{D}p_ix^i$ where $p_{i}$ is the probability that a randomly selected check equation of an instance of the outer codes is connected to $i$ generations. The minimum degree of $P(x)$ is two since any check equation should encounter at least two generations, and $\sum_{i=2}^{D}{p_{i}}=1$. Moreover, generations contributing in each check equation are considered to be distributed uniformly at random among all the generations. We refer to such outer codes as random linear outer codes. More details about selecting $R$ and designing $P(x)$ are left to Section~\ref{sec:optimization}.

In the second step of the encoding, SRLNC is performed on the partitioned outer coded packets in which the source repeatedly forms output packets to be sent to the receiver through the network. In particular, first for each output packet a generation index $j\in\{1,2,\dots,n\}$ is selected uniformly at random with replacement. Then, having selected a vector element $\beta \in (\mathrm{GF}(q))^g$ uniformly at random, an output packet is formed as the linear combination of the $g$ outer coded packets of the $j$th generation using $\beta$ as the coefficient vector. Finally, the output packet is transmitted through the network along with the index of the selected generation $j$, and the coefficient vector $\beta$.

At the intermediate nodes, coding is done by conventional SRLNC as in \cite{Efficient_Methods,RandomAnnex, Silva_Overlapping,tang13}. The complexity of encoding per output packet for Gamma network codes is $O(gs+\bar{d}gs(1-R)/R)$ at the source and $O(gs)$ at intermediate nodes, where $\bar{d}$ is the average degree of the outer code check nodes. This constant complexity per output packet thus gives rise to an overall linear encoding complexity in terms of the block length $K$.

\subsection{Decoding}\label{sec:decoding}

At the receiver, each received packet reveals a linear equation in terms of the outer coded packets of the corresponding generation in $\mathrm{GF}(q)$. The receiver constantly receives packets until it can form a full rank linear equation system for one of the generations. This generation is then decoded by Gaussian elimination. At this time, an iterative decoding process operating on the graph of Fig.~\ref{fig:decoding_graph} initiates.

Each iteration of this iterative decoding process is performed in two steps. In the first step, the \emph{edge-deletion decoding} step \cite{LT}, all the nodes corresponding to the outer coded packets of the recent full rank generations and their connecting edges are removed from the decoding graph. As a result, the degree of the check nodes of the outer code is reduced. Any outer code's check node reduced to degree one represents a dense linear equation in terms of the outer coded packets of the connected generation in $\mathrm{GF}(q)$. Thus, a dense linear equation is added to the linear equation system of the corresponding generation.

The second step follows by updating the linear equation system of the generations and performing Gaussian elimination for the full-rank generations. Any added dense linear equation increases the rank of the linear equation system of that generation by one with high probability if the alphabet size $q$ is large enough. As a result, there is a possibility that the updated generation becomes full rank and its packets could be recovered by Gaussian elimination.

The decoder now iterates between these two steps until either all the packets are recovered or no new packet could be recovered. If no new packet could be recovered, then the receiver receives more packets from the network so that it can resume the decoding. The decoding complexity of Gamma network codes is $O(g^2+gd+g\bar{d}(1-R)/R)$ operations per information packet which translates to a linear overall decoding complexity in terms of $K$.

%
%

\section{Asymptotic Analysis and Design} \label{sec:analysis}
In this section, we will study the average performance of our suggested Gamma network codes. The main goal of this study is to provide an analytical framework to formulate the effects of different code parameters on the average performance. As usual in the literature of modern coding, we will conduct this study under an asymptotic length assumption and derive density evolution equations for the iterative decoding process. Later, the finite-length performance of the example codes will be evaluated through computer simulations in Section~\ref{sec:results} along with the related discussions and remarks on finite-length issues.


As stated in Section~\ref{sec:decoding}, a successful decoding requires all of the generations to become full rank. Any received packet and any outer code's check node reduced to degree one add one dense linear equation to the equation system of the corresponding generation. For large $q$, adding one dense linear equation increases the rank of equation system by one with high probability. Thus, to analyze the decoding process, we are interested in tracking the evolution of the rank of the linear equation systems corresponding to different generations. To this end, in the following, we calculate the average fraction of generations whose equation systems are of rank $i,~i\in\{0,\dots,g\}$ at any instance during the decoding process.

Let the number of received encoded packets at some arbitrary state during the decoding be denoted by $rn$, where $0 \leq r$ is the normalized number of received encoded packets. Having a total of $r$ normalized number of received encoded packets, the decoder can form a system of linear equations in terms of the encoded packets in each generation. By a slight abuse of notations we will refer to the rank of such an equation system as the rank of its corresponding generation.

Let $R_{r,q}$ be the random variable representing the rank of a generation selected uniformly at random, when the normalized number of received encoded packets is equal to $r$ and the code alphabet is of size $q$. The following lemma whose proof is provided in App.~\ref{app:generation_rank_distribution}, gives the statistical structure of the generation rank distribution under very large $q$.

\begin{lem}\label{lemm:generation_rank_distribution}
\begin{align}
q\rightarrow \infty \Rightarrow R_{r,q}\overset{\mathfrak{D}}{\rightarrow}\mathcal{B}_{r,n},
\end{align}
where $\overset{\mathfrak{D}}{\rightarrow}$ denotes the convergence in distribution, and $\mathcal{B}_{r,n}$ is a random variable with the following truncated binomial probability distribution:
\begin{align}
\text{Pr}[\mathcal{B}_{r,n}=i]=\nonumber \begin{cases}
\binom{rn}{i}(\frac{1}{n})^{i}(\frac{n-1}{n})^{rn-i} & i=0,1,\dots ,g-1\\
1-I_{\frac{n-1}{n}}(rn-g+1,g) & i=g
\end{cases}.
\end{align}
Here $I_{\alpha}(m,\ell)$ is the regularized incomplete beta function defined as
\begin{align}
I_{\alpha}(m,\ell)=m\binom{m+\ell-1}{\ell-1}\int_{0}^{\alpha}{t^{m-1}(1-t)^{\ell-1}dt}.
\end{align}
\end{lem}\hfill$\square$

Since the main goal here is to study the average asymptotic performance, we assume that the value of $q$ is large enough to make the results of the previous lemma valid.


\begin{cor}
When the block length of the SRLNC goes to infinity, we have $n\rightarrow\infty$ and hence $R_{r,q}\overset{\mathfrak{D}}{\rightarrow}\mathcal{R}_r$, where $\mathcal{R}_r$ is a random variable with the following truncated Poisson distribution
\begin{align}\label{eq:GenRankDist}
\text{Pr}[\mathcal{R}_r=i]=\begin{cases}
\frac{e^{-r}r^{i}}{i!} & i = 0,1,\cdots,g-1\\
1-\frac{\Gamma_{g}(r)}{(g-1)!} & i=g
\end{cases},
\end{align}
where $\Gamma_{g}(r)$ is the incomplete Gamma function\footnote{Gamma network codes are named after the incomplete Gamma function since it plays a key role in their design.} given as
\begin{align}
\Gamma_{\alpha}(x)=(\alpha-1)!e^{-x}\sum_{i=0}^{\alpha}{\frac{x^i}{i!}}.
\end{align}
\end{cor}\hfill$\square$


Now that we have the probability distribution of the rank of a randomly selected generation at hand, we are interested to find the average number of generations of rank $i,~i\in\{0,1,\cdots,g\}$. The following lemma derives this quantity.

\begin{lem}\label{Lem:No_of_fullrank}
Let $E_{r}\{\cdot\}$ denote the expectation operator given that the normalized number of received packets is $r$. The average number of generations of rank $i$ is then given by
\begin{align}
E_{r}\left\{|\{\mathcal{G}|\text{rank}(\mathcal{G})=i\}|\right\}=n\text{Pr}[\mathcal{R}_r=i],
\end{align}
where $|A|$ denotes the cardinality of the set $A$.
\end{lem}\hfill$\square$

App.~\ref{app:no_of_fullrank} provides the proof of this lemma.

\subsection{Density Evolution Equations}\label{sec:density_evolution}

In the next step of our analysis, we study the growth in the average fraction of full rank generations during the decoding process, assuming that the packet reception has stopped at some arbitrary time. Let $r_{0}$ denote the normalized number of received encoded packets at this time.


The decoder has two sets of equations which could be used for decoding, namely the set of equations corresponding to the received encoded packets and the set of check equations available due to the outer code. Since the main goal in the design of SRLNC is to keep the decoding and encoding efficient, Gaussian Elimination is just performed within each generation, i.e., just performed on the set of equations which are all in terms of packets belonging to a single generation. For the check equations of the outer code, the decoder uses message-passing decoding (i.e., edge-deletion decoding) to reduce them to degree one.

At step zero of the iterative decoding process, where the normalized number of received encoded packets is $r_{0}$, the probability distribution of the rank of any randomly selected generation is given by (\ref{eq:GenRankDist}) as $\text{Pr}[\mathcal{R}_{r_0}=g]=1-\frac{\Gamma_g(r_0)}{(g-1)!}$. Therefore, the initial average fraction of full rank generations (i.e., before using any of the check equations in the decoding), is given by
\begin{align}\label{eq:initial_fullrank}
x_{0}=1-\frac{\Gamma_{g}(r_{0})}{(g-1)!}.
\end{align}

Having the developed mathematical framework at hand, it is now easy to track the average fraction of full rank generations as a function of the normalized number of received packets. In order to keep this simple formulation working for tracking the average fraction of full rank generations when the outer code comes to play in the decoding, we introduce the concept of \emph{effective} number of received packets. The aim of this definition is to translate the effect of check equations which are reduced to degree one into the reception of some imaginary packets from the network. This enables us to use the developed mathematical framework to track the average fraction of full rank generations as the decoding iterates between the edge-deletion decoder working on the outer code and the Gaussian elimination decoder which works inside each generation.

Now assume that after the $i$th iteration of the decoding for some $i\geq 0$, we have a certain fraction $x_i$ of full rank generations. Moreover, let $y_i$ denote the number of check equations of the outer code reduced to degree one at iteration $i$, which have not been reduced to degree one up to the end of iteration $i-1$. Each of these check equations now represents a dense equation in terms of the packets of one of the non-full rank generations. When $q$ is large enough, each of these equations will then increase the rank of its corresponding non-full rank generation by one, with high probability. However, as the selection of generations participating in each parity check equation in the outer code is done uniformly at random in the encoder, the effect of these equations on the decoding is equivalent to receive $y_i$ imaginary packets from the network all belonging to the non-full rank generations. Noticing that receiving more packets in the full rank generations also does not have any effect in the decoding process and does not change the fraction of full rank generations, we can easily model the effect of $y_i$ reduced degree-one parity check equations of the outer code by receiving $y_i/(1-x_i)$ imaginary packets from the network distributed uniformly at random over all the generations. We will refer to these $y_i/(1-x_i)$ imaginary packets as the \emph{effective} number of received packets at the beginning of iteration $i+1$. Moreover, we refer to the quantity
\begin{align}
z_{i+1}=n\Gamma^{-1}_{g}((1-x_i)(g-1)!)+y_i/(1-x_i),\nonumber
\end{align}
as the total effective number of received packets at the beginning of iteration $i+1$. According to Lemma~\ref{Lem:No_of_fullrank}, and the discussion above, the average fraction of full rank generations at iteration $i+1$ is given by
\begin{align}
x_{i+1}=1-\frac{\Gamma_{g}\left(\frac{z_{i+1}}{n}\right)}{(g-1)!}.\nonumber
\end{align}

Now consider the Tanner graph \cite{tanner81} of the outer code. Similar to the idea of density evolution \cite{DensityEvolution} and the intrinsic information transfer (EXIT) charts \cite{tenbrink99}, we track the density of full rank generations through the decoding iterations. In each iteration, in the first step all the edges connecting the full rank generations to the outer code's check nodes are removed. This reduces the degree of the check nodes. In the second step, each check node which is reduced to degree one adds a dense linear equation in terms of the packets of the connected generation to the coefficient matrix of that generation. The following theorem describes the evolution of the average fraction of full rank generations through the iterations of the decoding process.

\begin{thm}\label{thm:Lower_bound}
Let $r_{0}$ denote the normalized number of received packets, and $x_{i}$ for $i\geq 0$ denote the average fraction of full rank generations after iteration $i$ of decoding. Then the average effective number of received packets at iteration $i,~i\geq 1$ is given by
\begin{align}
ng(1-R)P'(x_{i-1})(1-x_{i}),\nonumber
\end{align}
Where $P'(\cdot)$ denotes the first order derivative of $P(x)$ and we have
\begin{align}\label{eq:evolution}
&x_{0} =  \nonumber 1-\frac{\Gamma_{g}(r_{0})}{(g-1)!},\\
&x_{i} = 1-\frac{\Gamma_{g}(r_{0}+g (1-R) P'(x_{i-1}))}{(g-1)!},~~i \geq 1.
\end{align}
\end{thm}

\begin{proof}
The initial average fraction of full rank generations $x_{0}$, could be calculated using (\ref{eq:initial_fullrank}). In the first iteration of the decoding, decoder removes the edges connecting the full rank generations connected to the outer code's check nodes. Thus, the probability of having a randomly selected check node reduced to degree one at this point is equal to
\begin{align}
\sum_{i=2}^{\infty}{p_{i}\binom{i}{1}(x_{0})^{(i-1)}(1-x_{0})}=P'(x_{0})(1-x_{0}).\nonumber
\end{align}
This is the probability of all except one of the generations participating in that check equation being full rank, and having that last one belong to the set of non-full rank generations. Such a check equation now reveals a dense equation in terms of packets of the only non-full rank generation connected to it and hence increases the rank of that non-full rank generation with high probability. Thus, the probability that a randomly selected check equation increases the rank of a non-full rank generation in iteration $1$ is
\begin{align}
P'(x_{0})(1-x_{0}).\nonumber
\end{align}

Moreover, as the total number of check equations is given by $N-K$, the average number of check equations which are now capable to increase the rank of a non-full rank generation is given as
\begin{align}
N(1-R)P'(x_0)(1-x_0) = ng(1-R)P'(x_0)(1-x_0).\nonumber
\end{align}

As discussed above, the effect of these $ng(1-R)P'(x_{0})(1-x_{0})$ equations on the generation rank growth is equivalent to the effect of $ng(1-R)P'(x_{0})$ dense equations distributed uniformly at random over all of the generations. Thus, we model the impact of iteration one of the edge-deletion by the reception of $ng(1-R)P'(x_{0})$ dense equations distributed uniformly at random over all of the generations. Then the average effective number of equations is $ng(1-R)P'(x_{0})$, or equivalently, the normalized average effective number of equations is
\begin{align}
g(1-R)P'(x_{0}).\nonumber
\end{align}

As all of the equations (i.e. effective check equations reduced to degree one, and equations corresponding to the received packets) which have been used in the coefficient matrices of the generations have a uniform distribution on the set of all generations, then the total average effective number of equations used throughout the decoding up to iteration one is equal to $r_{0}+g(1-R)P'(x_{0})$. Hence, similar to the calculation of $x_{0}$, we can calculate $x_{1}$ as
\begin{align}
x_{1} =  1-\frac{\Gamma_{g}(r_{0}+g(1-R)P'(x_{0}))}{(g-1)!}\nonumber.
\end{align}

Assuming the claim of Theorem~\ref{thm:Lower_bound} holds for all iterations from zero to $i$, we will now prove the claim for iteration $i+1$, and using mathematical induction we then conclude that the theorem holds for all iterations. Recall that we denote the average fraction of full rank generations at the end of iteration $i$ by $x_{i}$, and according to the assumption, the average effective normalized number of the total received packets up to the end of iteration $i$ is $r_{0}+g(1-R)P'(x_{i-1})$. Hence, according to the discussion above, the average fraction of check equations reduced to degree one after the edge deletion phase of iteration $i+1$ is given by $P'(x_{i})(1-x_{i})$. Since we have a total of $N-K=N(1-R)$ check equations, the number of check equations reduced to degree one is
\begin{align}\label{total_oreder_ones}
N(1-R)(1-x_{i})P'(x_{i}).
\end{align}

In order to calculate the average effective number of equations received at iteration $i+1$, we need to find the number of check nodes reduced to degree one at this iteration which have not been reduced to degree one in the previous iterations. Therefore, we need to deduct the average number of check nodes reduced to degree one up to the end of iteration $i$ which are still of degree one from (\ref{total_oreder_ones}). Hence, the total average effective number of received packets at this point is given by
\begin{align}
\nonumber nr_{0}+ng(1-R)P'(x_{i-1})+ &\frac{ng(1-R)}{(1-x_{i})}\left[P'(x_{i})(1-x_{i})-P'(x_{i-1})\frac{(1-x_{i})}{(1-x_{i-1})}\right]=\nonumber \\
& n\left[r_{0}+g(1-R)P'(x_{i})\right].\nonumber
\end{align}

Therefore, the average fraction of full rank generations at the end of iteration $i+1$ is given by
\begin{align}
x_{i+1}=1-\frac{\Gamma_{g}(r_{0}+g(1-R)P'(x_{i}))}{(g-1)!}.\nonumber
\end{align}

The claim of the theorem then holds for all iterations.
\end{proof}

\subsection{Decoding convergence and overhead}\label{sec:convergence}
Using Theorem~\ref{thm:Lower_bound}, a sufficient condition for successful decoding can be derived. Assume that packet reception is stopped after receiving enough packets to form $x_{0}n$ full rank generations, for some $x_{0}$ such that $0<x_{0}<1$. For large enough $q$ and $n$, the random linear outer code $\mathcal{C}$ with check degree distribution $P(x)$ then asymptotically guarantees successful decoding if
\begin{align}\label{eq:sufcon}
x < 1-\frac{\Gamma_g(r_{0}+g(1-R)P'(x))}{(g-1)!},~~x\in (x_{0},1),
\end{align}
where $r_0=\Gamma_g^{-1}\left((g-1)!(1-x_0)\right)$.

Note that to recover all of the encoded packets, $x$ should approach $1$ in (\ref{eq:sufcon}). But $x$ tends to $1$ when the argument of $\Gamma_g(\cdot)$ tends to infinity since $\Gamma_g(\cdot)$ is a strictly decreasing function lower bounded by zero. This means that $P'(x)$ should tend to infinity as $x$ tends to one. Since $x<1$ and $P(x)$ is a polynomial with positive coefficients, this is achieved only when the average degree of the outer code check nodes $\bar{d}$ tends to infinity\footnote{It can be shown that in this case the average degree should scale logarithmically with $n$.} which makes the per packet encoding and decoding complexities unbounded.

Motivated by the construction of Raptor codes \cite{Raptor} and to keep the complexities linear, we concatenate a high-rate linear block code $\mathcal{C}'$, which is called the \emph{pre-code}, with the random linear outer code $\mathcal{C}$. For this purpose, we use a weakened random linear outer code $\mathcal{C}$ of rate $R$ with a small constant $\bar{d}$. A constant $\bar{d}$ means that a fraction of the generations will remain uncovered. The pre-code $\mathcal{C}'$ is then responsible to recover the remaining fraction of the generations. As a result, if we choose $\bar{d}$ and $P(x)$ such that \begin{equation}\label{eq:convergence_condition}
x<1-\frac{\Gamma_g(r_{0}+g(1-R)P'(x))}{(g-1)!},~~x\in (x_{0},1-\delta),
\end{equation}
given a small $\delta>0$, then a capacity-achieving pre-code of rate $R'=1-\delta$ can recover the remaining $\delta$ fraction of generations.

Due to the concatenation of the pre-code, encoding of Gamma network codes should now be done in three steps. In the first step, a file consisting of $K'$ packets is encoded via $\mathcal{C}'$ with rate $R'$ to give a block of $K=K'/R'$ packets. In the next step, encoding this block by the outer code $\mathcal{C}$ of rate $R$ gives a block of $N=K/R=ng$ outer coded packets. The final step consists of the conventional RLNC. The number of information packets is given by $K'=ngR'R=ng(1-\delta)R$. The receiver is able to successfully decode all of the information packets after receiving $r_0n$ encoded packets from the network. As a result, the average reception overhead of this coding scheme is given by
\begin{align}
\nonumber\epsilon = \frac{r_0n-K'}{K'}&=\frac{r_0}{g(1-\delta)R}-1\\
&=\frac{\Gamma_g^{-1}\left((g-1)!(1-x_0)\right)}{g(1-\delta)R}-1\label{eq:overhead}
\end{align}

Considering these, the asymptotic convergence properties of Gamma network codes can be summarized as follows. For a Gamma network code with a linear random outer code of rate $R$ and check degree distribution $P(x)$, if (\ref{eq:convergence_condition}) is satisfied for some $x_0$ and $\delta$, then the Gamma network code can asymptotically recover all of the information packets with an average reception overhead of (\ref{eq:overhead}) using a linear capacity-achieving pre-code\footnote{The pre-code can be a high-rate right-regular low-density parity-check code (LDPC) designed for the binary erasure channel (BEC)\cite{Shokrollahi99}.} of rate $1-\delta$.

Moreover, in the asymptotic regime, the variance of the fraction of recovered generations approaches zero as shown by \cite{Concentration_Kaplan,Concentration_Flatto}. Hence, the average behavior is expected to be observed with high probability.

We conclude this section by an example. For the heuristic outer code design proposed in \cite{Mahdaviani12}\footnote{This heuristic design is based on the assumption that minimizing the overhead can be achieved to a great extent by designing the code such that $x_0 = 1/n$ \cite{Mahdaviani12}.} with $g=25$, we have outer code rate $R=0.6351$, and precode rate $R'=0.9701$ and
\begin{equation}\label{eq:heuristic_px}
P(x) = \sum_{i=2}^{D^*}\frac{1}{i(i-1)}x^i+\frac{1}{D^*}x^{(D^*+1)},
\end{equation}
where $D^*=33$. The evolution of $x_i$ during the decoding process as predicted by (\ref{eq:evolution}) is plotted in Fig.~\ref{fig:EXIT_example} for $x_0=0.10$. Also, the $45$-degree line is plotted. We call this the decoding evolution chart. The point where the evolution chart gets closed, i.e., intersects the $45$-degree line, is equal to $1-\delta$. As depicted in Fig.~\ref{fig:EXIT_example}, $1-\delta$ is very close to one for this example. The predicted average asymptotic reception overhead given by (\ref{eq:overhead}) is then $18.83\%$.

\begin{figure}
\centering
\includegraphics[width=\columnwidth]{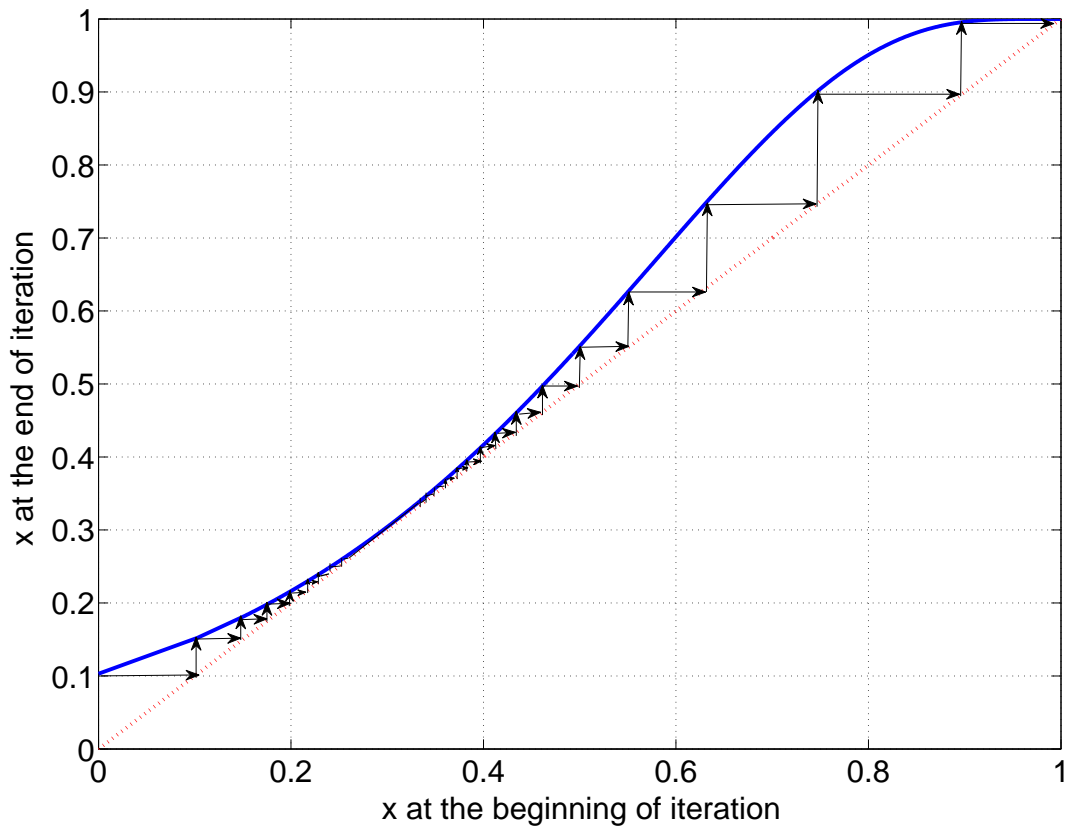}
\caption{The decoding evolution chart for the Gamma network code with the heuristic check degree distribution of \cite{Mahdaviani12}. Using these parameters, the lowest $x_0$ by which the evolution chart is open is $x_0=0.10$.}\label{fig:EXIT_example}
\end{figure}


The concept of decoding evolution chart is very similar to the EXIT chart introduced in \cite{tenbrink99}. This concept has already been used to derive many techniques for the analysis and optimization in many modern coding problems, and is proved to be a very powerful tool with many applications. In network coding, however, no similar concept has been introduced prior to this work. In the following sections we will describe some examples of applying the decoding evolution chart for optimization of the proposed Gamma network codes in asymptotic and finite block lengths as examples of this tool.

\section{Outer Code Optimization} \label{sec:optimization}
The previous section provided us with the tools for the asymptotic analysis of the decoding of Gamma network codes as well as their decoding convergence and reception overhead calculation. Now that this analytical formulation is at hand, we can use it to design good Gamma network codes. The goal of this design process is to find a combination of the parameters of the outer code and the pre-code, namely the rate of the outer code $R$, the check degree distribution $P(x)$, and the rate of the pre-code $R'$, which gives the minimum reception overhead.

For this purpose, we are seeking solution to the following optimization problem:
\begin{align}
\min_{{R, P(x), x_0,\delta}}\epsilon=&\min_{R, P(x), x_0,\delta}\frac{\Gamma_g^{-1}\left((g-1)!(1-x_0)\right)}{g(1-\delta)R}-1\label{eq:optimization}\\
\nonumber&\quad\mathrm{subject~to:}~~(\ref{eq:convergence_condition})~\mathrm{holds}\\
\nonumber&\quad\quad\qquad\qquad\;\;\sum_{i=2}^{D}{p_{i}}=1\\
\nonumber&\qquad\quad\qquad\qquad0\le p_i\le1.
\end{align}
Solving this optimization problem analytically is not easy since some of the parameters inherently depend on each other through the non-linear constraint (\ref{eq:convergence_condition}). Thus, we use numerical methods to find solutions to this optimization problem.

First notice that for a fixed $R$ and $P(x)$, for any given $x_0$ one can find $\delta$ by using the convergence condition (\ref{eq:convergence_condition}). Also, since $0< x_0<1$ and $0<R<1$, for any fixed $P(x)$ one can make a fine grid and do a search over $(x_0,R)$ and minimize $\epsilon$ and find the best combination of $x_0$, $R$, and $\delta$. Finally, solution to (\ref{eq:optimization}) can be found by fixing maximum degree $D$ and searching over $P(x)$ and finding the best set of $x_0$, $R$, and $\delta$ for each $P(x)$.

Searching over the space of $P(x)$ can be done by global numerical optimization techniques. Here, to speed up the process, we use the gradient descent method to find various local minima and then choose the best answer. Our results are not guaranteed to be equal to the global minimum but as we will show in our examples, the decoding evolution chart for the optimized codes gets extremely close to the $45$-degree line which suggests that our results should be very close to the global answer.

Assuming that the generation size is $g=25$, asymptotically optimized Gamma network codes are found for various values of the maximum check degree $D$ by solving (\ref{eq:optimization}). The parameters of these codes are reported in Table~\ref{tb:optimization}. Selecting $D=2$ is equivalent to an all degree-2 check degree distribution. In this case, the check degree distribution is fixed and the rest of the parameters are optimized (code $\mathcal{C}_1$ in the table). The reception overhead under this code is $\epsilon=11.43\%$. The evolution chart of the decoding of this code is plotted in Fig.~\ref{fig:EXIT_optimized}.

As evident from the results of Table~\ref{tb:optimization}, increasing $D$ from $2$ to $30$ decreases the reception overhead from $11.43\%$ to $2.60\%$. This is because increasing $D$ allows larger average degrees for $P(x)$ and hence the closing point of the evolution chart gets closer to $x=1$. Also, note that the reception overhead does not change significantly for $D>15$ since the closing point $1-\delta$ is already very close to $1$ and larger average degrees does not change $1-\delta$ and hence the overhead significantly. The decoding evolution charts for $\mathcal{C}_4$ and $\mathcal{C}_6$ which are optimized under $D=15$ and $D=30$, respectively, are also depicted in Fig.~\ref{fig:EXIT_optimized}.

Note that in the optimized degree distributions of Table~\ref{tb:optimization}, only check nodes of degree $2$ and $D$ have significant weights, with most of the weight on degree $2$. Having a large weight on degree-$2$ check nodes is useful since it maximizes the participation of the outer code's check nodes. Degree-$2$ check nodes start to contribute early at the beginning of the decoding but since they provide low connectivity in the decoding graph, they fail to be useful eventually when the fraction of recovered packets grow. Low connectivity in the graph make some segments of the graph unrecoverable since the decoding process cannot spread to all segments. This leads to a significant increase in the reception overhead. As an example, in the all degree-$2$ code $\mathcal{C}_1$, the outer code participates in the decoding sooner than the other codes with larger $D$ (compare $x_0=0.049$ with the rest) but fails to contribute in the decoding when the fraction of full rank generations gets larger (by having a smaller $1-\delta$) and a lower rate pre-code is needed to finish the decoding. Large-degree check nodes, on the contrary to degree-$2$ check nodes, provide good coverage in the graph but cannot participate early in the decoding since the low fraction of recovered packets is unlikely to reduce them to degree one. Consequently, there should be a balance between degree $2$ and higher degrees. This balance is usually achieved by putting a large weight on degree $2$ and the rest of the weight on the largest allowed degree.


\begin{figure}
\centering
\includegraphics[width=\columnwidth]{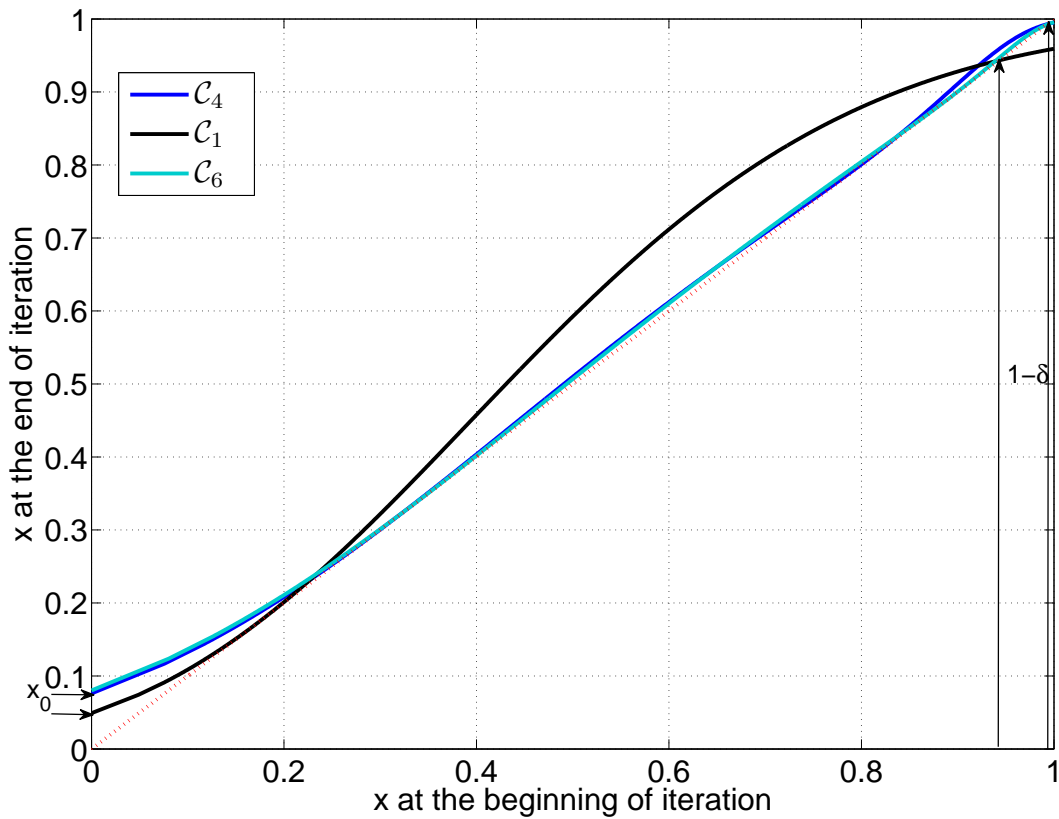}
\caption{The decoding evolution chart for the optimized Gamma network codes $\mathcal{C}_1$, $\mathcal{C}_4$, $\mathcal{C}_6$. The parameters of these codes are reported in Table~\ref{tb:optimization}.}\label{fig:EXIT_optimized}
\end{figure}

\begin{table}[t!]
\caption{Optimized check degree distributions $P(x)=\sum_ip_ix^i$ under $g=25$ for various maximum check degrees $D$} \label{tb:optimization}
\begin{center}
\begin{tabular}{|c|c|c|c|c|c|c|}
\hline
& $\mathcal{C}_1$ & $\mathcal{C}_2$ & $\mathcal{C}_3$ & $\mathcal{C}_4$ & $\mathcal{C}_5$ & $\mathcal{C}_6$\\
\hline
$D$    & $2$      & $5$      & $10$     & $15$     & $20$      & $30$    \\
\hline
$p_2$  & $1.0000$ & $0.7860$ & $0.8788$ & $0.9226$ & $0.9184$  & $0.9162$\\
\hline
$p_3$  &          &          &          &          & $0.0011$  &         \\
\hline
$p_4$  &          &          &          & $0.0004$ &           & $0.0004$\\
\hline
$p_5$  &          & $0.2140$ &          & $0.0004$ &           & $0.0028$\\
\hline
$p_6$  &          &          &          &          & $0.0012$  & $0.0069$\\
\hline
$p_7$  &          &          &          &          & $0.0071$  & $0.0065$\\
\hline
$p_8$  &          &          & $0.0002$ &          & $0.0138$  & $0.0092$\\
\hline
$p_9$  &          &          & $0.0003$ & $0.0005$ & $0.0082$  & $0.0095$\\
\hline
$p_{10}$ &        &          & $0.1207$ & $0.0010$ & $0.0036$  & $0.0075$\\
\hline
$p_{11}$ &        &          &          &          & $0.0005$  & $0.0068$\\
\hline
$p_{12}$ &        &          &          &          & $0.0003$  & $0.0055$\\
\hline
$p_{13}$ &        &          &          &          &           & $0.0032$\\
\hline
$p_{14}$ &        &          &          & $0.0048$ &           &         \\
\hline
$p_{15}$ &        &          &          & $0.0703$ &           &         \\
\hline
$p_{19}$ &        &          &          &          & $0.0004$  &         \\
\hline
$p_{20}$ &        &          &          &          & $0.0455$  &         \\
\hline
$p_{26}$ &        &          &          &          &           & $0.0007$\\
\hline
$p_{27}$ &        &          &          &          &           & $0.0006$\\
\hline
$p_{28}$ &        &          &          &          &           & $0.0002$\\
\hline
$p_{29}$ &        &          &          &          &           & $0.0002$\\
\hline
$p_{30}$ &        &          &          &          &           & $0.0239$\\
\hline
\hline
$x_0$  & $0.0490$ & $0.1100$ & $0.0885$ & $0.0762$ & $0.0782$& $0.0802$\\
\hline
$R$    & $0.6600$ & $0.7342$ & $0.7228$ & $0.7163$ & $0.7192$& $0.7216$\\
\hline
$1-\delta$& $0.9433$ & $0.9746$ & $0.9912$ & $0.9910$ & $0.9910$ & $0.9911$\\
\hline
$\epsilon$ & $11.43\%$ & $6.62\%$ & $3.64\%$ & $2.75\%$ & $2.65\%$ & $2.60\%$\\
\hline
\end{tabular}
\end{center}
\end{table}


The minimum reception overhead can be further decreased by increasing the generation size $g$. For example, reception overheads of $\epsilon=2.17\%$ and $\epsilon=1.92\%$ can be achieved under $D=15$ when $g=50$ and $g=75$, respectively ($\mathcal{C}_7$ and $\mathcal{C}_8$ in Table~\ref{tb:optimization2}). This reduction in the minimum reception overhead is however achieved at the expense of added encoding and decoding complexities\footnote{It is worth mentioning that the complexity still remains linear and only the coefficient increases.}.

\begin{table}[t!]
\caption{Optimized check degree distributions $P(x)=\sum_ip_ix^i$ under $D=15$ for generation sizes $g=50$ and $g=75$.} \label{tb:optimization2}
\begin{center}
\begin{tabular}{|c|c|c|}
\hline
& $\mathcal{C}_7$ & $\mathcal{C}_8$\\
\hline
$D$    & $15$      & $15$    \\
\hline
$p_2$  & $0.9260$ & $0.9303$  \\
\hline
$p_3$  & $0.0007$ &           \\
\hline
$p_5$  & $0.0002$ & $0.0001$  \\
\hline
$p_6$  & $0.0002$ &           \\
\hline
$p_7$  & $0.0006$ & $0.0005$  \\
\hline
$p_8$  & $0.0010$ & $0.0002$  \\
\hline
$p_9$  & $0.0005$ & $0.0003$  \\
\hline
$p_{10}$ &$0.0001$&           \\
\hline
$p_{11}$ &$0.0001$& $0.0002$  \\
\hline
$p_{12}$ &$0.0001$& $0.0002$  \\
\hline
$p_{13}$ &$0.0018$&           \\
\hline
$p_{14}$ &$0.0018$& $0.0025$  \\
\hline
$p_{15}$ &$0.0669$& $0.0658$  \\
\hline
\hline
$g$    & $50$     & $75$     \\
\hline
$x_0$  & $0.0831$ & $0.0853$  \\
\hline
$R$    & $0.8008$ & $0.8374$  \\
\hline
$1-\delta$& $0.9911$ & $0.9911$  \\
\hline
$\epsilon$ & $2.17\%$ & $1.92\%$  \\
\hline
\end{tabular}
\end{center}
\end{table}

\section{Numerical Results and Robust Design} \label{sec:results}
In this section, we investigate the performance of Gamma network codes constructed based on the results of the previous section. In particular, we investigate the reception overhead and decoding failure probability trade-off of Gamma network codes in practical settings and compare them with the other existing SRLNC schemes. We also discuss issues regarding their finite-length performance and provide robust and improved designs.

\subsection{Simulation setup}
The pre-code $\mathcal{C}'$ should be a capacity-achieving code which does not incur extra overhead. To this end, we use the right-regular capacity-achieving binary low-density parity-check (LDPC) codes designed for the binary erasure channel (BEC) \cite{Shokrollahi99}. The check nodes of the pre-code, as opposed to the check nodes of the outer code, impose parity-check equations directly on the encoded packets. Decoding of the pre-code and the outer code is done jointly. As a result, during the decoding any pre-code check node reduced to degree one recovers a new coded packet. This updates the linear equation system for the generation to which the recovered packets belong by removing the new recovered coded packets from them. This reduces the linear equation system of those generations to the non-recovered packets. Since the number of unknowns are reduced, there is a possibility that the non-recovered packets of the updated generations can be recovered by Gaussian elimination. It is also worth mentioning that since the pre-code is a high-rate code, the degrees of its check nodes are usually very large. Thus, they are reduced to degree one and hence help the decoding process only at the final stages of the decoding when a large fraction of the coded packets are recovered.

Using a finite alphabet size $q$ and having designed a pre-code $\mathcal{C}'$ of rate $R'$, a random linear outer code $\mathcal{C}$ of rate $R$, and considering encoded packets block length of $N=ng$, where $g$ is the generation size, we calculate the average reception overhead of the coding scheme by Monte Carlo simulation, i.e., $\bar{\epsilon}=E[(N_r-K')/K']$ where $N_r$ is the number of received packets required for successful decoding. To achieve the trade-off between decoding failure probability and overhead, we simulate the system for a large number of blocks and calculate the empirical complementary cumulative distribution function of the overhead.

\subsection{Numerical results}\label{subsec:results_B}
For a finite-length setup, we set the alphabet size $q=256$, generation size $g=25$, and the number of generations $n=67$, which gives an encoded packet block of $N=1675$ blocks. Using the parameters of the asymptotically optimized code $\mathcal{C}_4$ from Table~\ref{tb:optimization}, we have $R=0.7163$ and $K=RN=1200$. As a result, $N-K=475$ check nodes are produced based on the optimized degree distribution $P(x)$ in Table~\ref{tb:optimization}. For the pre-code, we use a right-regular binary LDPC code of rate $R'=0.97$. This rate is selected slightly lower than the asymptotically optimized rate of $0.991$ due to the fact that there is a gap between the finite-length performance and asymptotic performance of capacity-achieving LDPC codes\footnote{In practice, the best pre-code rate giving the minimum reception overhead can be selected by Monte Carlo simulation.}. The number of information packets will then be $K'=1164$. The average reception overhead achieved by Monte Carlo simulation is $\bar{\epsilon}=10.82\%$. The decoding failure probability versus the reception overhead is plotted in Fig.~\ref{fig:overhead1}.

\begin{figure}
\centering
\includegraphics[width=\columnwidth]{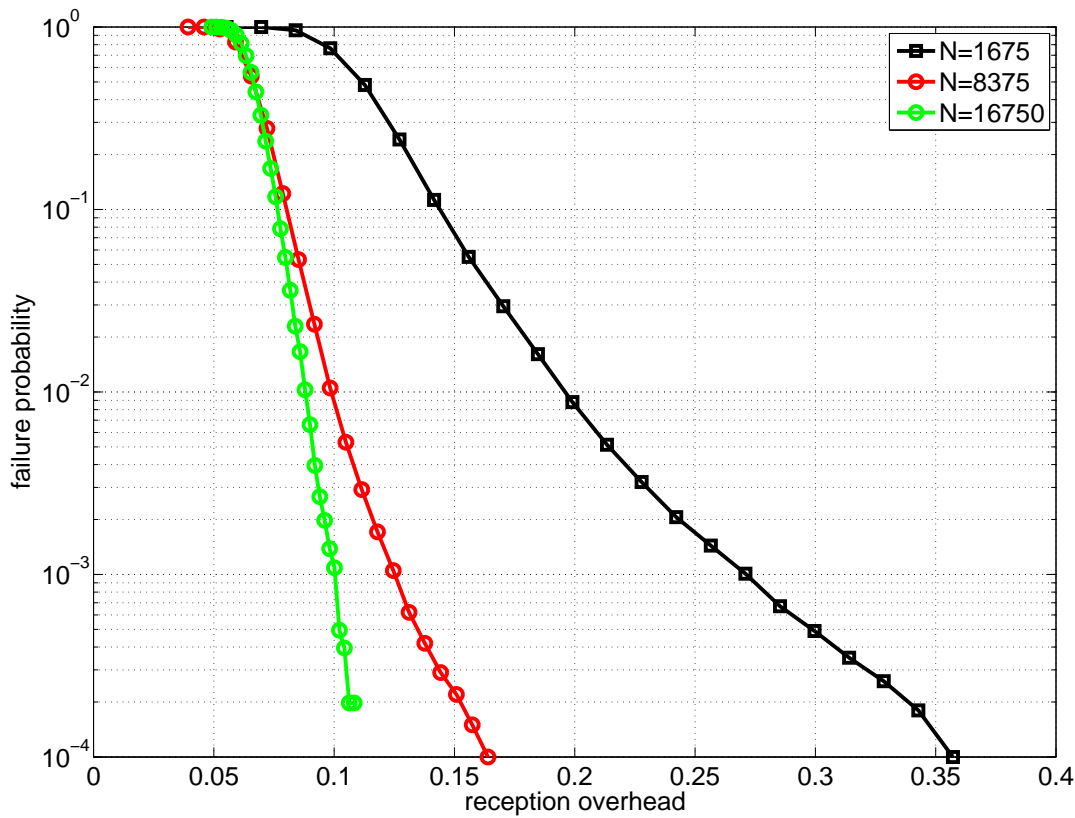}
\caption{Failure probability versus reception overhead for Gamma network codes of different lengths. For all lengths, $g=25$, $q=256$, and $P(x)$ and $R$ are equal to those of $\mathcal{C}_4$ in Table~\ref{tb:optimization}. The encoded packet block lengths are $N=1675$, $N=8375$, and $N=16750$, with pre-codes of rates $R'=0.97$, $R'=0.98$, and $R'=0.98$, respectively. The pre-codes are all binary right-regular LDPC codes. The average reception overheads of these schemes achieved by Monte Carlo simulation are reported in Table~\ref{tb:average_overhead2}.}\label{fig:overhead1}
\end{figure}


We expect improvements in the performance by increasing the block length of the code. To show this, we have also constructed codes with larger block lengths, namely $N=8375$ and $N=16750$. The pre-codes are right-regular binary LDPC codes of rate $R'=0.98$, $g=25$, and the rest of the parameters are the same as those of $\mathcal{C}_4$ in Table~\ref{tb:optimization}. The performances of these codes are depicted in Fig.~\ref{fig:overhead2} which show improvement with regard to the code with length $N=1675$. Table~\ref{tb:average_overhead2} includes the average reception overhead $\bar{\epsilon}$ achieved for these constructions. As $N$ increase, $\bar{\epsilon}$ gets closer to the asymptotic overhead reported in Table~\ref{tb:optimization} for $\mathcal{C}_4$. It is worth mentioning that by using $R'=0.98$ (instead of $R'=0.991$), the asymptotic achievable reception overhead will be $\epsilon=3.90\%$ which is very close to the empirical result obtained at $N=16750$.

We are also interested in investigating how the performance of optimized Gamma network codes varies with alphabet size $q$. Table~\ref{tb:average_overhead2} also includes the average reception overhead achieved by Monte Carlo simulation for the codes constructed with $q=2$ and $q=16$. It is clear that $\bar{\epsilon}$ increases by decreasing $q$ from $256$ to $2$. Also, note that the $\bar{\epsilon}$ achieved by $q=16$ is very close to that by $q=256$. Thus, in practice $q=16$ should normally be enough.



\begin{table}[t!]
\caption{Average overhead for optimized Gamma network codes constructed based on the parameters of $\mathcal{C}_4$ with $g=25$} \label{tb:average_overhead2}
\begin{center}
\begin{tabular}{|c|c|c|c|c|c|c|c|c|c|}
\hline
$N$ & \multicolumn{3}{|c|}{$1675$} & \multicolumn{3}{|c|}{$8375$} & \multicolumn{3}{|c|}{$16750$}\\
\hline
$q$ & $2$ & $16$ & $256$ & $2$ & $16$ & $256$ & $2$ & $16$ & $256$ \\
\hline
$\bar{\epsilon}$ & $21.33\%$ & $11.71\%$ & $10.82\%$ & $18.09\%$ & $7.23\%$ & $6.64\%$ & $16.54\%$ & $5.96\%$ & $5.57\%$ \\
\hline
\end{tabular}
\end{center}
\end{table}

\subsection{Robust Gamma Network Codes} \label{sec:robust_practical}
As evident from the failure-probability versus overhead performance of the asymptotically optimized  Gamma network code with $N=1675$ and $N=8375$ in Fig.~\ref{fig:overhead2}, achieving very low probabilities of failure increases the reception overhead significantly. In other words, the curve is not very steep and there exists an error floor.

The reason for the existence of error floor for highly optimized Gamma network codes can be described as follows. The decoding evolution chart of highly optimized codes is normally very close to the $45$-degree line which makes their opening very narrow, e.g., see Fig.~\ref{fig:EXIT_optimized}. As stated, the evolution chart which is based on (\ref{eq:evolution}) and (\ref{eq:convergence_condition}), predicts the average performance of asymptotic Gamma network codes. When the convergence condition (\ref{eq:convergence_condition}) is satisfied, receiving $r_0n$ packets from the network is enough to trigger a chain reaction in the decoding such that the asymptotic Gamma network code recovers all of the encoded packets without getting stuck and receiving any more packets from the network.

When finite-length codes are used, however, the performance deviates from the average performance expected for the asymptotic regime. As a result, for the finite-length case, the decoder might get stuck several times during the decoding and can only continue after receiving enough packets from the network to form a new full rank generation. Getting stuck in the early stages of decoding when the fraction of recovered packets is small does not increase the reception overhead significantly since the new received packets most likely belong to the non-full rank generations and with high probability they increase the rank of their corresponding generation. However, getting stuck when the fraction of recovered packets is large (equivalent to the upper portion of the decoding evolution chart), normally leads to a significant increase in the reception overhead as most of the new received packets belong to the already full rank generations. The event of getting stuck in the final stages of the decoding happens with low probability but it incurs a large overhead. This is why the error floor exists for these codes in the finite-length cases.

The above discussion suggests that having an asymptotic decoding evolution chart which is widely open at its upper portion leads to codes with smaller error floors since this decreases the probability of getting stuck at points where the fraction of recovered packets is large. Thus for a robust design, asymptotic Gamma network codes can be optimized under an additional constraint to have decoding  evolution charts widely open in the upper portion. This can be done by modifying the convergence constraint to
\begin{equation}\label{eq:robust_constraint}
x<1-\frac{\Gamma_g(r_{0}+g(1-R)P'(x))}{(g-1)!},~~x\in (x_{0},1-\delta_0]
\end{equation}
and
\begin{equation}\label{eq:robust_constraint2}
x<1-\frac{\Gamma_g(r_{0}+g(1-R)P'(x))}{(g-1)!}-\Delta,~~x\in (1-\delta_0,1-\delta'),
\end{equation}
for some $\Delta>0$ and $x_0<1-\delta_0<1-\delta'$, and modify the minimization problem to
\begin{align}
&\min_{R, P(x), x_0,\delta'}\frac{\Gamma_g^{-1}\left((g-1)!(1-x_0)\right)}{g(1-\delta')R}-1.\label{eq:optimization2}\\
\nonumber&\quad\mathrm{subject~to:}~~(\ref{eq:robust_constraint})~\mathrm{and}~ (\ref{eq:robust_constraint2})~\mathrm{hold}\\
\nonumber&\quad\quad\qquad\qquad\;\;\sum_{i=2}^{D}{p_{i}}=1\\
\nonumber&\quad\qquad\qquad\qquad0\le p_i\le1
\end{align}
Notice that the closing point ($1-\delta'$) of the decoding evolution chart given by the modified convergence conditions (\ref{eq:robust_constraint}) and (\ref{eq:robust_constraint2}) is used in the robust optimization problem. This closing point is not the closing point of the true asymptotic evolution chart of the decoding because of the margin $\Delta$. The true asymptotic convergence condition and evolution chart are still given by (\ref{eq:convergence_condition}). After solving (\ref{eq:optimization2}), the pre-code rate is found to be $R'=1-\delta'$ and the overhead of this concatenation will then be
\[\epsilon=\frac{\Gamma_g^{-1}\left((g-1)!(1-x_0)\right)}{g(1-\delta')R}-1.\]

The parameters of two such robust codes designed by setting $\Delta=0.03$, $1-\delta_0=0.8$ for $\mathcal{C}_9$, and $\Delta=0.01$, $1-\delta_0=0.9$ for $\mathcal{C}_{10}$, and solving (\ref{eq:optimization2}) are given in Table~\ref{tb:robust}. Fig.~\ref{fig:EXIT_robust} depicts the decoding evolution chart of $\mathcal{C}_9$. The average reception overhead achieved for a finite-length construction of $\mathcal{C}_9$ with $R'=0.9644$, $N=1675$, and $q=256$ is $\bar{\epsilon}=13.45\%$. The performance is also depicted in Fig.~\ref{fig:robust} where it is shown that the error floor can be decreased using the above robust optimization method. This is achieved at the expense of a slight increase in the average reception overhead. Fig.~\ref{fig:robust} also contains the performance of a robust Gamma network code with $N=8375$ constructed using the parameters of $\mathcal{C}_{10}$ which also shows decrease in error floor. In this case, $\bar{\epsilon}=6.88\%$

\begin{table}[t!]
\caption{Robust optimized codes designed by solving (\ref{eq:optimization2}) and assuming $D=15$} \label{tb:robust}
\begin{center}
\begin{tabular}{|c|c|c|}
\hline
& $\mathcal{C}_9$ & $\mathcal{C}_{10}$\\
\hline
$p_2$  & $0.8443$ & $0.9074$\\
\hline
$p_3$  &  & \\
\hline
$p_4$  & $0.0006$ & \\
\hline
$p_5$  & $0.0006$ & \\
\hline
$p_6$  & $0.0005$ & \\
\hline
$p_7$  & $0.0005$ & \\
\hline
$p_8$  & $0.0024$ & \\
\hline
$p_9$  & $0.0022$ & $0.0006$ \\
\hline
$p_{10}$& $0.0347$ & \\
\hline
$p_{11}$& $0.0265$ & $0.0024$\\
\hline
$p_{12}$& $0.0453$ & \\
\hline
$p_{13}$& $0.0249$ &\\
\hline
$p_{14}$& $0.0081$ & \\
\hline
$p_{15}$& $0.0094$ & $0.0896$\\
\hline
\hline
$g$    & $25$  & $25$  \\
\hline
$x_0$  & $0.0838$ & $0.0777$\\
\hline
$R$    & $0.7046$ & $0.7144$ \\
\hline
$1-\delta'$& $0.9644$ & $0.9820$\\
\hline
$1-\delta_0$& $0.8000$ & $0.9000$\\
\hline
$\Delta$& $0.0300$ & $0.0100$\\
\hline
$\epsilon$ & $8.55\%$ & $4.20\%$\\
\hline
\end{tabular}
\end{center}
\end{table}

\begin{figure}
\centering
\includegraphics[width=\columnwidth]{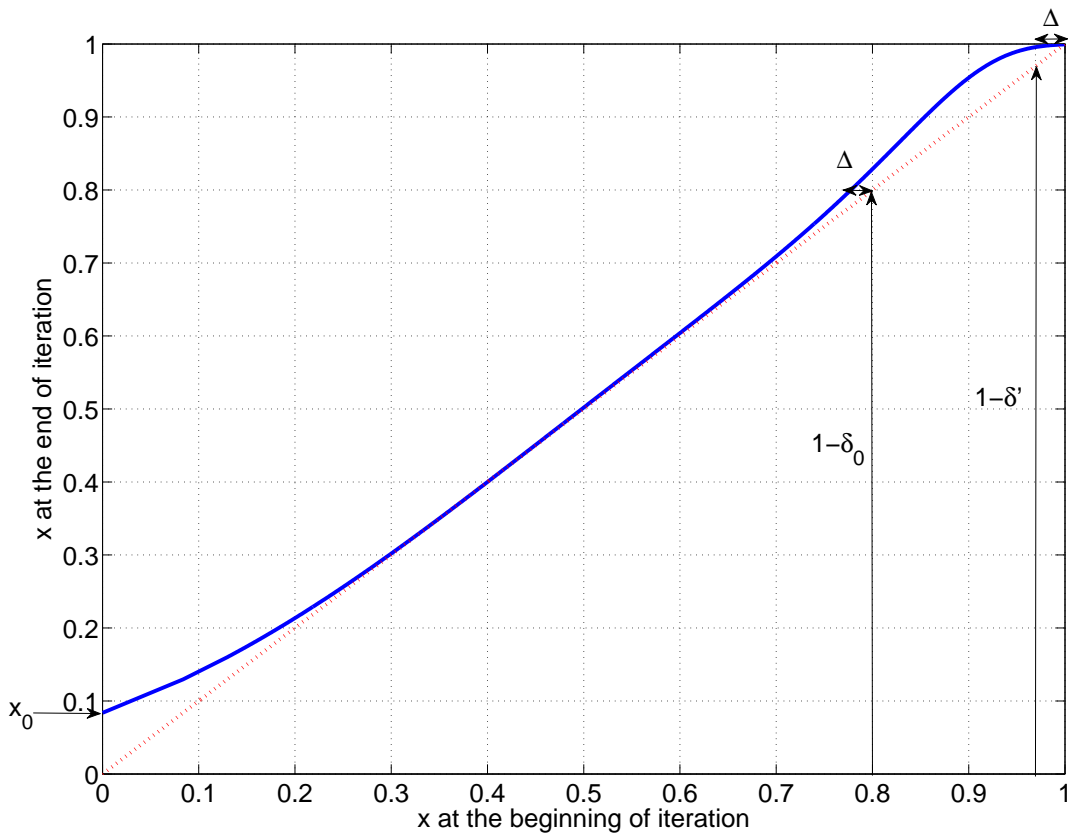}
\caption{The decoding evolution chart for the robust optimized Gamma network code $\mathcal{C}_9$ with parameters $D=15$, $\Delta=0.03$, and $1-\delta_0=0.8$. Notice that this code has a wide opening at the upper portion of its evolution chart compared to that of $\mathcal{C}_4$.}\label{fig:EXIT_robust}
\end{figure}

\begin{figure}
\centering
\includegraphics[width=\columnwidth]{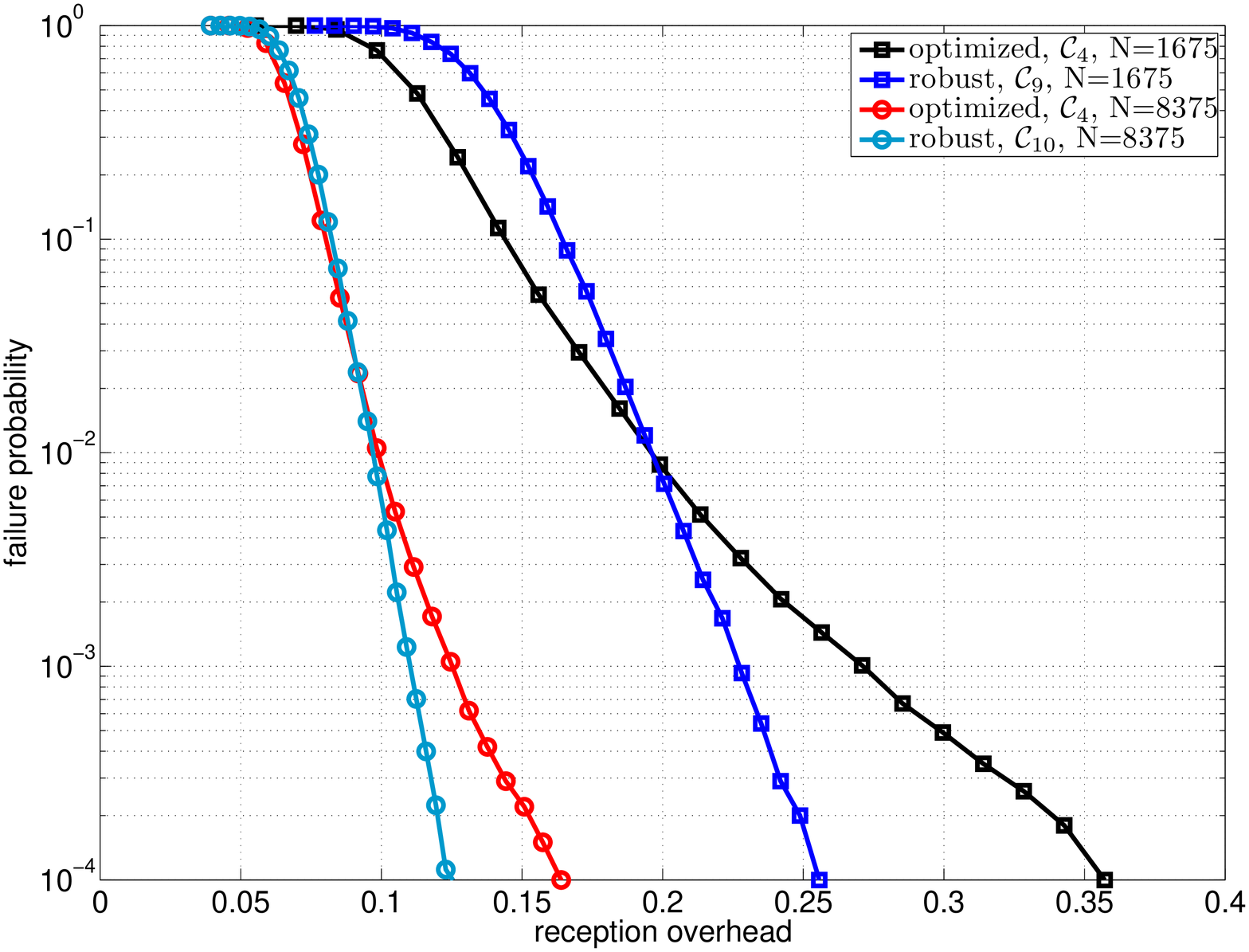}
\caption{Failure probability versus reception overhead for $\mathcal{C}_4$, a network Gamma code optimized for minimum average reception overhead and for $\mathcal{C}_9$ and $\mathcal{C}_{10}$ which are robust Gamma network codes designed based on the method of Section~\ref{sec:robust_practical}. The parameters of the robust codes are reported in Table~\ref{tb:robust}. It is clear that the error floor is improved under the robust design compared to the optimized design of $\mathcal{C}_4$.}\label{fig:robust}
\end{figure}

\subsection{Comparison with other SRLNC schemes}\label{Comparison_results}
In this section, we compare Gamma network codes with the other existing linear-complexity SRLNC schemes  \cite{Efficient_Methods,RandomAnnex,Silva_Overlapping,tang13}. The schemes of \cite{Efficient_Methods,RandomAnnex,Silva_Overlapping} lack exact analysis and design methods and are normally designed heuristically. However, there exists an asymptotic analysis and design method for EC codes based on expander graph arguments\cite{tang13}. As a result, we are also able to compare the performance of optimized Gamma network codes with EC codes in the asymptotic regime as well \cite{tang13}.

Assuming infinite block length and large enough $q$, EC codes can be designed with overheads $\epsilon=6.62\%$ and $\epsilon=5.50\%$ for $g=25$ and $g=50$, respectively. From Tables~\ref{tb:optimization} and \ref{tb:optimization2}, we see that Gamma network codes achieve average overheads of $\epsilon=2.60\%$ and $\epsilon=2.17\%$ for $g=25$ and $g=50$, respectively. This shows that Gamma network codes outperform EC codes.

For a finite-length comparisons, we use the following schemes. In the case of SRLNC with an outer LDPC code as a separate block \cite{Efficient_Methods}, the optimal rate for $N=8375$ and $g=25$ is found by search at $R=0.90$ \cite{Mahdaviani12} which gives rise to $K=7538$. For the Random Annex codes of \cite{RandomAnnex}, the optimal annex size is found to be $11$ for $g=25$ and $N=8375$ which gives rise to $R=0.56$ and hence $K=4690$ packets. For the overlapping SRLNC scheme of \cite{Silva_Overlapping}, the parameters of an optimal diagonal grid code are found to be $(5000,25,335)$ with $\theta=7$ which is equivalent to having a repetition outer code of rate $R=0.5970$. These two schemes do not use any pre-code. In the case of EC codes, the optimal overlap size is $16$ giving rise to $R=0.68$ for the overlapped code and the pre-code is a right-regular LDPC code of rate $R'=0.98$. For our Gamma network codes, we use the parameters of $\mathcal{C}_{10}$ with $N=8375$ and an LDPC pre-code of rate $R'=0.98$.

The average reception overhead achieved under these cases have been reported in Table~\ref{tb:average_overhead}. Fig.~\ref{fig:overhead2} also compares these schemes with our optimized Gamma network codes in terms of failure probability-overhead trade-off. As evident from these results, the optimized Gamma network code outperforms all the other existing outer coded SRLNC schemes. However, we will see later that even further improvement to the performance of Gamma network codes is also possible.

\begin{table}[t!]
\caption{Average overhead for different linear-complexity SRLNC schemes with outer code, $N=8375$, $g=25$, $q=256$} \label{tb:average_overhead}
\begin{center}
\begin{tabular}{|c|c|}
\hline
Code & $\bar{\epsilon}$ \\
\hline
SRLNC with LDPC & $41.07\%$\\
\hline
Random annex code & $31.69\%$\\
\hline
Diagonal grid code & $29.77\%$\\
\hline
EC code & $7.83\%$\\
\hline Gamma network code, robust $\mathcal{C}_{10}$ & $6.88\%$\\
\hline Gamma network code, optimized $\mathcal{C}_4$ & $6.45\%$\\
\hline
\end{tabular}
\end{center}
\end{table}

\begin{figure}
\centering
\includegraphics[width=\columnwidth]{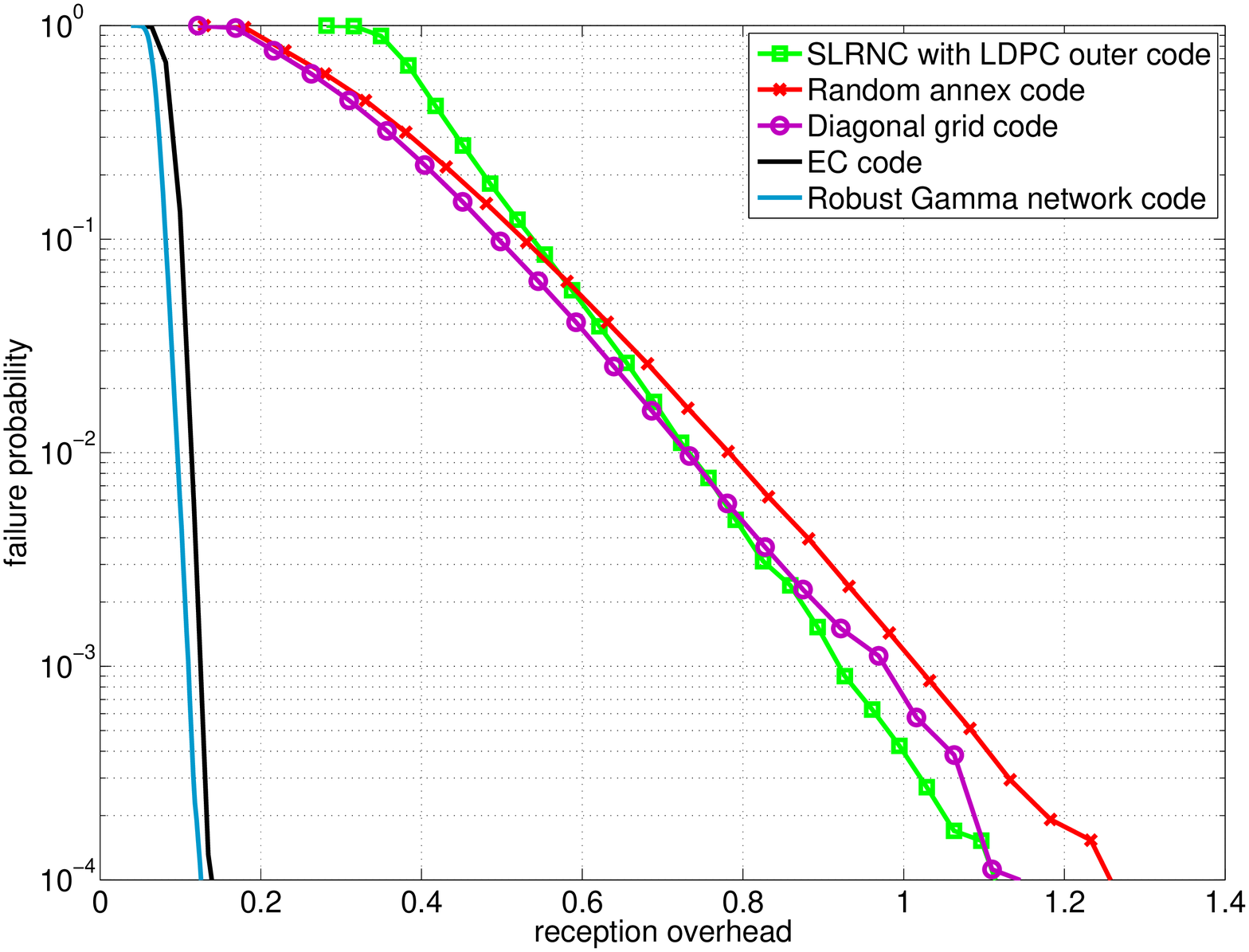}
\caption{Failure probability versus reception overhead for different SRLNC schemes with outer code. The encoded packet block length for all of these constructions is $N=8375$ packets with $g=25$ and $q=256$. The average reception overheads of these schemes achieved by Monte Carlo simulation are reported in Table~\ref{tb:average_overhead}. It is evident that the Gamma network code outperforms all other schemes. }\label{fig:overhead2}
\end{figure}

\section{Encoding and Improved Designs}\label{sec:improved_design}
\subsection{Encoding}\label{sec:improved_design_encoding}
We stated in Section~\ref{sec:encoding} that the check nodes of the outer code impose parity-check constraints on the dense linear combinations of all of the packets of their adjacent generations. This is different from how a check node of a conventional linear code imposes constraints directly on the connected packets. Thus, the encoding process of the outer code in Gamma network codes is different from conventional linear codes.

To achieve linear encoding complexity per block length, encoding the outer coded packets can be done as follows. Assume that we have an outer code of rate $R=K/N$ with check degree distribution $P(x)$ and generation size $g$.
\begin{enumerate}
        \item We construct an instance of the ensemble of Tanner graphs specified by $P(x)$ which connects $n=N/g$ generations to $N-K$ check nodes. We then call the number of check nodes connected to each generation $G_i$ the degree of that generation $d_{G_i}$.
        \item The $n$ generations are sorted based on their degrees in a descending order such that $d_{G_i}\ge d_{G_{i+1}}$, $1\le i\le n-1$.
        \item The $K$ pre-coded packets constitute $K$ outer coded packets as in a systematic code. These packets are distributed into the $n$ generations based on the following rules:
            \begin{enumerate}
                \item Generation $i$ receives $m_i=g-[{d_{G_i}}/{\bar{d}}]$ where $\bar{d}$ is the empirical average degree of the check nodes and $[\cdot]$ denotes rounding to the nearest integer.
                \item It is ensured that $\sum_{i=1}^nm_i=K$.
            \end{enumerate}
       \item Generation $G_i$ now contains $m_i$ packets $\{u^i_1,\dots,u^i_{m_i}\}$. Set $i=1$.
       \item For $G_i$, we select $g-m_i$ number of check nodes among the $d_i$ check nodes connected to $i$ with the highest check degrees. The set of these check nodes is denoted by $\mathcal{M}(G_i)$.
       \item We generate $g-m_i$ parity packets $\{u_{m_{i+1}}^i,\dots,u_{g}^i\}$ as
        \begin{equation}
            u_j^i=\sum_{k\in\mathcal{N}(\mathrm{c})}\sum_{l=1}^{n_k}\alpha_l^ku_l^k,~j\in\{m_{i+1},\dots,g\}
        \end{equation}
       where $\mathrm{c}\in\mathcal{M}(G_i)$, $\mathcal{N}(\mathrm{c})$ denotes the set of generations connected to $\mathrm{c}$, $n_{G_k}$ denotes the number of packets currently available in $G_k$, and $\alpha$ are random coefficients from $\mathrm{GF}(q)$.
       \item If $i=n$ stop. Else set $i:=i+1$ and go to step 5.
\end{enumerate}
This algorithm ensures that the number of packets which participate in the random linear combinations are maximized.


\subsection{Improved designs}\label{sec:improved_design_improved}
It is also possible to impose parity-check constraints directly on the outer coded packets as in a conventional linear code instead of their dense linear combinations. In this case, instead of (\ref{eq:generation_based_checks}), the parity-check equation represented by check node $\mathrm{c}$ is given by
\begin{equation}\label{eq:packet_based_checks}
\sum_{i\in\mathcal{N}(\mathrm{c})}u^{(i)}=0.
\end{equation}
Then, $p_i$ in $P(x)=\sum p_ix^i$ will represent the probability that any given check node be connected to $i$ outer coded packets. The decoding process for such a code should be modified since any outer code's check node which is reduced to degree one, similar to the pre-code's check nodes, now recovers an outer coded packet instead of adding a dense linear equation to its corresponding generation. If the new to-be-recovered packet has not already been decoded, it can now be removed from the linear equation system of its belonging generation. Since the linear equation system has dense coefficient vectors, rank will be preserved with high probability and with less number of unknowns now there exists a possibility that the equation system can be solved.

In the analysis of the proposed outer code of Section~\ref{sec:encoding}, we assumed that every check node reduced to degree one increases the rank of its corresponding generation by one. This assumption is not valid in general when check nodes are imposed on outer coded packets since a reduced degree-one check node may be connected to an already recovered packet. Nevertheless, if the outer code is constructed in such a way that each outer coded packet is connected to at most one check node, then a reduced degree-one check node always recovers a new outer coded packet. This requires the average degree of check nodes to be upper bounded as $\bar{d}<1/(1-R)$. Under this assumption, the results of Theorem~\ref{thm:Lower_bound} and the convergence condition (\ref{eq:convergence_condition}) are still valid.

For a given $P(x)$, $R$, and $r_0$, as the convergence condition (\ref{eq:convergence_condition}) predicts, the outer code can recover $1-\delta$ fraction of the generations. At this point, as given by (\ref{total_oreder_ones}), there will be $\delta N(1-R)P'(1-\delta)$ check nodes of degree one which belong to the remaining $\delta$ fraction of generations. This means that $\delta N(1-R)P'(1-\delta)$ are already recovered among the $\delta N$ packets belonging to the remaining non-full-rank generations. Thus, the rate of the pre-code required to recover the remaining outer coded packets is
\begin{equation}\label{eq:precode_packet_level}
R'= 1-\frac{N\delta-\delta N(1-R) P'(1-\delta)}{N} = 1-\delta+\delta(1-R)P'(1-\delta).
\end{equation}
This then gives an average reception overhead of
\begin{equation}\label{eq:overhead_packet_level}
\epsilon= \frac{\Gamma_g^{-1}\left((g-1)!(1-x_0)\right)}{gR(1-\delta+\delta(1-R)P'(1-\delta))}-1.
\end{equation}
Since $\delta (1-R)P'(1-\delta)>0$, then the average reception overhead predicted by (\ref{eq:overhead_packet_level}) will be smaller than (\ref{eq:overhead}). Thus, the average reception overhead of the case where check nodes impose constraints directly on packets is smaller than the construction of Section~\ref{sec:encoding}.

The average reception overhead of (\ref{eq:overhead_packet_level}) defines a new objective for the optimization problem given in (\ref{eq:optimization}). Notice that the new objective can change the results of the optimization significantly when $1-\delta$ is not very close to one. For example, when $D=2$, the new optimization problem gives $x_0=0.0540$, $R=0.6800$, and $1-\delta=0.9172$ leading to $R'=0.9658$ and average reception overhead of $\epsilon=6.77\%$. This shows a significant reduction in overhead compared to the overhead of $11.43\%$ reported in Table~\ref{tb:optimization} for $\mathcal{C}_1$.

\emph{Remark:} Note that in the case of improved Gamma network codes, setting the maximum degree of the outer code to be $D=2$ as in the above mentioned example reduces the Gamma network codes to the SRLNC with a repetition outer code, or in other words, SRLNC with overlapping generations. Moreover, it is interesting to note that in this case, the optimal outer code rate for the Gamma network code is calculated to be $R=0.6800$ which is in a very close agreement with the EC codes design \cite{tang12,tang13} as mentioned in Section~\ref{Comparison_results}. However $D =2$ is not the optimal choice for Gamma network codes and its performance can be further improved by increasing the maximum degree and hence outperforms all the previously existing SRLNC schemes including the EC codes as depicted in Fig.~\ref{fig:overhead3}.

Given the parameters of the optimized Gamma network code $\mathcal{C}_4$, we have constructed finite-length Gamma network codes whose outer code have packet-level check nodes. Notice that in this case, since $1-\delta$ is very close to one, the results of the new optimization will not be significantly different from those of Table~\ref{tb:optimization}, except for $R'$ and the average reception overhead $\epsilon$. We have chosen the new pre-codes used in the simulations to be LDPC codes of rate $R'=0.99$. Table~\ref{tb:average_overhead3} shows the average reception overhead achieved under these codes. Fig.~\ref{fig:overhead3} also compares the performance of the code of length $N=8375$ with the robust Gamma network code and EC codes of Fig.~\ref{fig:overhead2}. As the figure shows, code designed with packet-level check nodes and the LDPC pre-code of rate $R'=0.99$ outperforms the other designs.

\begin{table}[t!]
\caption{Average overhead for optimized Gamma network codes constructed based on the parameters of $\mathcal{C}_4$ with $g=25$ and $q=256$. The check nodes of the outer code impose parity-check equations directly on the packets. The pre-codes are right-regular LDPC codes of rate $R'=0.99$.} \label{tb:average_overhead3}
\begin{center}
\begin{tabular}{|c|c|c|c|}
\hline
$N$ & $1675$ & $8375$ & $16750$ \\
\hline
$\bar{\epsilon}$ & $10.30\%$ & $5.75\%$ & $5.18\%$\\
\hline
\end{tabular}
\end{center}
\end{table}

\begin{figure}
\centering
\includegraphics[width=\columnwidth]{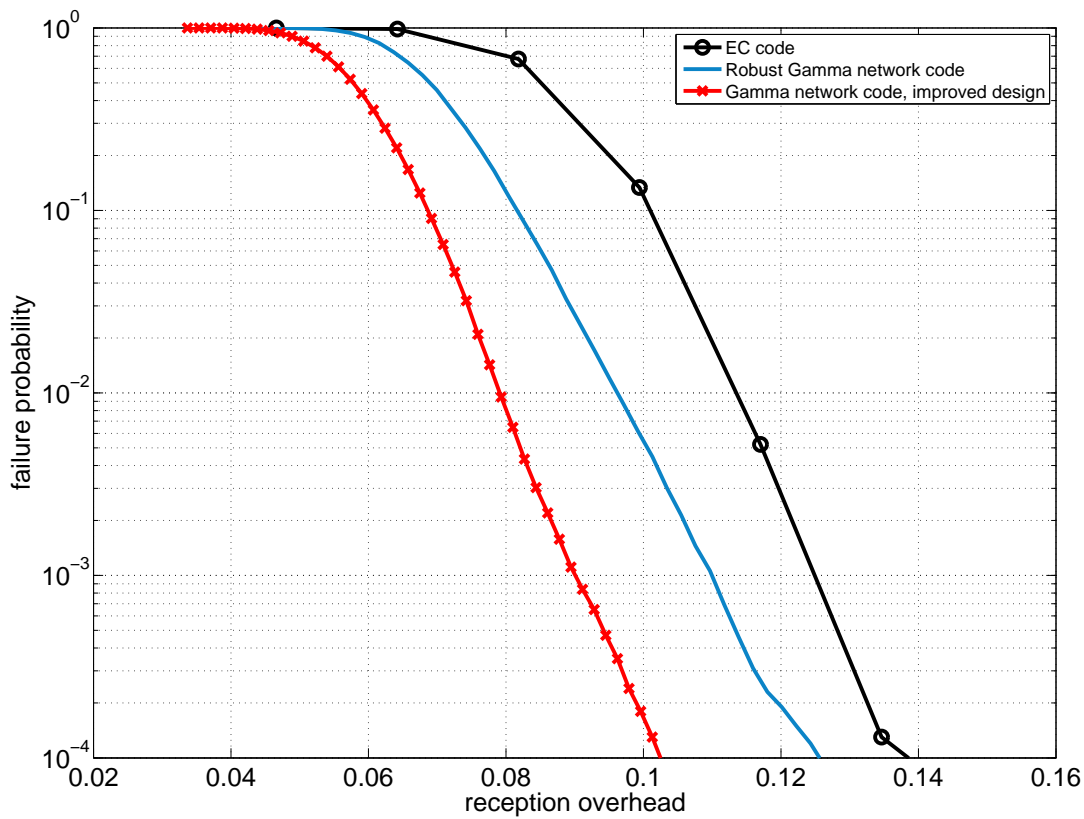}
\caption{Failure probability versus reception overhead comparison between the optimized Gamma network code with packet-level outer code check nodes and the robust Gamma network code and EC codes of Fig.~\ref{fig:overhead2}.}\label{fig:overhead3}
\end{figure}

\section{Conclusion}\label{sec:conclusion}
In this paper, we introduced and studied a new family of overhead-efficient SRLNC schemes called Gamma network codes. The introduced scheme was based on incorporating a linear outer code in the SRLNC construction. We then presented an analytical framework by formulating density evolution equations for the analysis and design of asymptotic Gamma network codes. Using the proposed analysis method, we presented an optimization technique to design minimum overhead Gamma network codes and obtain their fundamental limits. We followed our studies with numerical results and showed that Gamma network codes outperform all the other existing SRLNC schemes. Realizing that highly optimized Gamma network codes exhibit error floors in finite-lengths, we proposed a robust design method to lower the error floors. We finally discussed their encoding method and also introduced improved designs.

\section*{Acknowledgement}
The authors wish to thank Hossein Bagheri for many valuable discussions which has influenced parts of this work.

\appendices
\section{Proof of Lemma \ref{lemm:generation_rank_distribution}} \label{app:generation_rank_distribution}
Let $B_{r,n}$ be a random variable and its outcome be the number of encoded packets received for a randomly selected generation, when the normalized number of received encoded packets is $r$. Since it is assumed that the probability of a randomly selected received encoded packet belonging to a certain generation has a uniform distribution on the set of all the generations, $B_{r,n}$ has a binomial probability distribution as
\begin{align}
\text{Pr}[B_{r,n}=i]=\binom{rn}{i}\left(\frac{1}{n}\right)^{i}\left(\frac{n-1}{n}\right)^{rn-i},~i=0,1,\dots,rn.
\end{align}

To complete the proof, we use the result of theorem 3.1 in \cite{RaptorQ}. This theorem states that if $A_{m\times n}$ is a matrix in which each element is chosen independently and uniformly at random from $\mathrm{GF}(q)$, then for $n\leq m$
\begin{align}\label{rankdef}
\text{Pr}[\text{rank}(A)<n]\leq \frac{1}{(q-1)q^{m-n}}
\end{align}
Now, for $i\in\{0,1,\cdots,g-1\}$ we have,
\begin{align}
\text{Pr}[R_{r,q}=i] = \sum_{j=i}^{rn}{\text{Pr}[R_{r,q}=i|B_{r,n}=j]\text{Pr}[B_{r,n}=j]}\nonumber
\end{align}
But according to (\ref{rankdef}) it is easy to see that as $q\rightarrow\infty \Rightarrow \text{Pr}[R_{r,q}=i|B_{r,n}=j]\rightarrow 0$ for all $i<j$. Furthermore, as $\sum_{i=0}^{j}{\text{Pr}[R_{r,q}=i|B_{r,n}=j]}=1$, we have $q\rightarrow \infty \Rightarrow {\text{Pr}[R_{r,q}=j|B_{r,n}=j]}\rightarrow 1$. Hence, $q\rightarrow \infty \Rightarrow {\text{Pr}[R_{r,q}=i]}\rightarrow {\text{Pr}[B_{r,n}=i]},~i=0,1,\cdots,g-1$.

In addition,
\begin{align}
\text{Pr}[R_{r,q}=g] = \sum_{j=g}^{rn}{\text{Pr}[R_{r,q}=i|B_{r,n}=j]\text{Pr}[B_{r,n}=j]}.\nonumber
\end{align}
Again as $q\rightarrow\infty \Rightarrow \text{Pr}[R_{r,q}=i|B_{r,n}=j]\rightarrow 0$ for all $i<g\leq j$ as a direct corollary of (\ref{rankdef}), we have $q\rightarrow \infty \Rightarrow {\text{Pr}[R_{r,q}=g|B_{r,n}=j]}\rightarrow 1$, for all $g\leq j$. Finally, we have
\begin{align}
q\rightarrow \infty \Rightarrow {\text{Pr}[R_{r,q}=g]}&\rightarrow \sum_{j=g}^{rn}{\text{Pr}[B_{r,n}=j]}\nonumber \\&=1-\text{Pr}[B_{r,n}\leq g-1]\nonumber \\&=1-I_{\frac{n-1}{n}}(rn-g+1,g) \nonumber
\end{align}

\section{Proof of Lemma \ref{Lem:No_of_fullrank}} \label{app:no_of_fullrank}
Take $A_{r,i,j},~0\leq r,~i \in \{1,\cdots,g\}, j \in \{1,\cdots,n\}$, as the event that the $j$th generation is of rank $i$ based on the received encoded packets' corresponding equations, when the normalized number of received encoded packets is $r$. Then we have
\begin{align}
&\mathbb{E}_{r}\{|\{G|G\in \mathcal{G},~\text{rank}(G)=i\}|\}\nonumber \\&=\mathbb{E}\lbrace\sum_{j=1}^{n}{I_{A_{r,i,j}}(\omega)}\rbrace,~\omega \in \Omega,
\end{align}
where $\mathcal{G}$ is the set of all generations, $\Omega$ is the set of all possible outcomes of the packet reception, and $I_{A}$ is the indicator function of the event $A$, i.e.,
\begin{align}
I_{A}(\omega)=\begin{cases} \nonumber
 1, & \text{if }\omega\in A \\
 0, & \text{if }\omega\notin A
\end{cases}.
\end{align}

Now, although the random variables $I_{A_{r,i,j}}$ are correlated, using the linearity of the expectation we have
\begin{align}
\mathbb{E}_{r}\{|\{G|G\in \mathcal{G},~\text{rank}(G)=i\}|\}&=\sum_{j=1}^{n}{\mathbb{E}\{I_{A_{r,i,j}}(\omega)\}}\nonumber \\ &=n\text{Pr}[\mathcal{R}_r=i]. \nonumber
\end{align}

\bibliographystyle{IEEEtran}
\bibliography{IEEEabrv,NetworkCoding_bib}

%

\end{document}